\pdfoutput=1
\ifdefined\TOMS
\documentclass[acmsmall,review,anonymous]{acmart}
\else
\documentclass[acmsmall,nonacm]{acmart}
\fi

\acmJournal{TOMS}

\usepackage{verbatim}
\usepackage{xspace}
\usepackage{amsmath}
\usepackage{url}
\usepackage{bm}
\usepackage{booktabs} %

\makeatletter
\newcommand{\ifanon}[2]{\if@ACM@anonymous#1\else#2\fi}
\makeatother

\def\..{\,\mathpunct{\ldotp\ldotp}} %
\makeatletter\let\tabinput\@@input\makeatother
\newcommand{\Z}{\mathbf Z}
\newcommand{\Q}{\mathbf Q}
\newcommand{\R}{\mathbf R}
\newcommand{\F}{\mathbf F}
\newcommand{\Zb}{\Z_b}

\newcommand{\mwc}[1][]{\texttt{MWC#1}\xspace}
\newcommand{\gmwc}[1][]{\texttt{GMWC#1}\xspace}
\newcommand{\xoshiropp}[1][]{\texttt{xoshiro#1++}\xspace}
\newcommand{\xoroshiropp}[1][]{\texttt{xoroshiro#1++}\xspace}
\DeclareMathOperator{\ord}{ord}

\begin{document}

\bibliographystyle{ACM-Reference-Format}
\acmJournal{TOMS}

\title{Scrambled Multiply-with-Carry Pseudorandom Number Generators}

\author{Sebastiano Vigna}
\orcid{0000-0002-3257-651X}
\affiliation{%
  \institution{Universit\`a degli Studi di Milano}
  \department{Dipartimento di Informatica}
  \city{Milano}
  \country{Italy}
}
\email{vigna@acm.org}

\begin{abstract}
Multiply-with-carry (MWC) generators are very fast pseudorandom number
generators whose quality depends on the choice of a multiplier, and is usually
assessed by applying the spectral test to an associated multiplicative
congruential generator. The special form of the modulus is shown to impose
structural limitations on the lattice structure of every MWC generator, which
the generalized MWC generators of Goresky and Klapper mitigate; in the tests
reported, a scrambler made of a single xor removes them. Since these
limitations make the usual worst-case criterion useless, candidates are ranked by a
weighted geometric mean of their figures of merit, which does not depend on
their normalization, and multipliers are suggested for generators with $64$ to $256$
bits of state. Along the way, some errors in the literature are corrected. In the
end, MWC192X is proposed: a $64$-bit generator with $192$ bits of state that
passes all known tests, and emits an output in half a nanosecond. Analogous
32-bit generators are proposed for GPUs.
\end{abstract}



\maketitle

\section{Introduction}

Congruential generators are among the oldest and best understood pseudorandom
number generators (PRNGs): their outputs have a lattice structure that can be
analyzed precisely with number-theoretic tools~\cite{KnuACPII}. This paper is
about a family of very fast generators whose analysis reduces to that of
congruential generators: generators with carry.

Marsaglia and Zaman~\cite{MaZNCRNG} introduced in 1991 \emph{add-with-carry}
(AWC) and \emph{subtract-with-borrow} (SWB) generators, which extend lagged
Fibonacci generators by propagating a carry (or borrow). Tezuka, L'Ecuyer, and
Couture~\cite{TLCLSACSBRNG} showed that they are, in disguise, multiplicative
congruential generators (MCGs) with a huge modulus, whose lattice structure is
very bad in dimensions larger than the lag (see also~\cite{CoLLSCLCSRASG}), a
defect that L\"uscher's RANLUX generator~\cite{LusPHQRNGLFTS} mitigates by
discarding most of the outputs.

In 1994, Marsaglia~\cite{MarYAR} proposed \emph{multiply-with-carry} (MWC)
generators, in which the sum of two lagged values is replaced by a linear
combination of lagged values. Couture and L'Ecuyer~\cite{CoLLRCURNG,CoLDPMCRNG}
showed that they, too, simulate an MCG: when only one coefficient $a$ is
nonzero, if $b$ is the base (typically, $2^{32}$ or $2^{64}$) and $k$ is the
lag, the modulus is $m=ab^k-1$ and the multiplier is $b^{-1}\bmod m$. They also
introduced \emph{complementary} MWC generators, with modulus $ab^k+1$, later
popularized by Marsaglia~\cite{MarRNG}. Independently, Klapper and
Goresky~\cite{KlGTASR,KlGFSRSCM} studied \emph{feedback with carry shift
registers} (FCSRs), the binary analogue of MWC generators, using $2$-adic
numbers. Goresky and Klapper~\cite{GoKEMCRNGMP} later merged the two points of
view, and proposed a \emph{generalized} MWC (GMWC) generator, in which the new
digit is multiplied by a coefficient $a_0$ coprime with the
base~\cite{KlXAFSR}, so that the modulus can have a much less constrained form.

On $64$-bit architectures with fast $64\times64\to128$-bit multiplications, MWC
generators are a very interesting choice: an MWC generator with $128$ bits of
state needs one full multiplication and one addition, and has period about
$2^{127}$; when the lag is larger than one, the multiplications of consecutive
outputs are independent, and modern processors execute them in parallel
(Section~\ref{sec:speed}). \ifanon{Since version 25.0, Erlang/OTP includes a $59$-bit MWC
generator~\cite{OTPrand}}{The author has been distributing MWC and GMWC
generators for a few years,\footnote{\url{https://prng.di.unimi.it/}} and since
version 25.0 Erlang/OTP includes a $59$-bit MWC generator whose multiplier was
selected by the author and Raimo Niskanen~\cite{OTPrand}}
(Section~\ref{sec:small}).

The quality of an MWC generator depends on the multiplier $a$, and the natural
tool to assess it is the \emph{spectral test}~\cite{CoMFAURNG,KnuACPII} of the
associated MCG, which computes, in each dimension $t$, the maximal distance
$d_t$ between successive hyperplanes over all families of equidistant parallel
hyperplanes covering the points produced by the generator. The distances are
normalized into \emph{figures of merit} $f_t\in(0\..1]$, which are summarized by
a single criterion, such as their minimum, the \emph{worst-case figure of merit}
$M_T$~\cite{FiMEAMCRNG,LEcTLCGDSGLS}, or the \emph{harmonic score} of Steele and
Vigna~\cite{StVCESGMCPNGjr}. The form of the modulus of an MWC generator,
however, forces a bad lattice structure: in dimensions up to $k+1$ the figures of
merit are fixed by the parameters, and in dimension $k+2$ they are, as we show,
smaller than $10^{-3}$ for $b=2^{64}$, so the worst-case figure of merit is
determined by a single, structurally bad dimension.

Our main result makes this limitation explicit: whatever the lag $k$ and the
multiplier $a<b$, the vectors of outputs $(x_n,x_{n+1},x_{n+k+1})/b$ of an MWC
generator lie, up to an error smaller than $1/b$ in each coordinate, in a family
of equidistant parallel planes at least about $1/(1.28\sqrt b)$ apart
(Section~\ref{sec:lacunary}; the argument is similar to that used by L'Ecuyer and
coauthors for multiple recursive generators with few nonzero
coefficients~\cite{LETDLRNGRM,LESLSSCMRRNG}). The resulting bounds on the figures of merit are
nearly attained by our best multipliers and, together with a generalization,
explain all the structural limitations we observe.
A three-dimensional
birthday-spacings test detects the defect in a few seconds for $b=2^{32}$, but
for $b=2^{64}$ our test uses $3$ to $5$\,TB of output, and needs an implementation
that can process them independently of the available memory
(Section~\ref{sec:detect}).

GMWC generators
with lag larger than one have a similar defect, with planes about $b^{-2/3}$
apart, whereas for GMWC generators with lag one we find no structural
limitation in the dimensions we consider. In any case, the
scrambler X, which generalizes to GMWC generators that of Thomas's MWC64X
generator (the exclusive or of each digit with the corresponding carry),
removes, in our tests, the defect and every other form of lattice structure we
could detect (Section~\ref{sec:x}).

Since the worst-case figure of merit is useless, we rank candidates by a
weighted geometric mean of their figures of merit: by a classical
characterization from social choice theory~\cite{dAGSWFIC}, this is essentially
the only ranking that depends neither on the normalization of the figures of
merit nor on the other candidates (Section~\ref{sec:criterion}). We report the
results of a search over more than $52$~million candidates for nine
combinations of base, lag, and generator type, and suggest multipliers.
Our final suggested generator is \mwc[192X]: in our benchmarks (Section~\ref{sec:speed}),
\mwc[192X] and \mwc[256X] are the fastest scalar generators known to \ifanon{us}{the author}
that have a known period and jump functions, and pass all the statistical tests
we are aware of (Section~\ref{sec:recommendations}). We suggest also $32$-bit
generators for GPUs.

We give a self-contained presentation of (G)MWC generators based on $b$-adic
numbers, following Goresky and Klapper~\cite{GoKEMCRNGMP}, with simple seeding
instructions for GMWC generators, and we correct three errors
in~\cite{GoKEMCRNGMP}: the formula for the carry is incorrect when $a_0\neq 1$,
three of the twelve proposed parameter sets do not have maximum period, and the
exponential representation of the output has the wrong sign (a related error
appears in~\cite{GoKASRS}). Finally, our title nods to the scrambled linear
generators of Blackman and Vigna~\cite{BlVSLPNG}: our generators, too, are made
of an engine and a scrambler, but the engine is based on congruential arithmetic
rather than on $\F_2$-linear maps, and, dually, our scrambler uses $\F_2$-based rather
than congruential arithmetic.

The examined candidate multipliers and their spectral data are available at
\anon{\url{https://vigna.di.unimi.it/MWC.tar.zst}}, the code of the search, of
the selection, and of our generators at
\anon{\url{https://github.com/vigna/MWCPRNG}}, and the harnesses and the speed
tests of Sections~\ref{sec:speed} and~\ref{sec:gpu} at
\anon{\url{https://github.com/vigna/prng-speed}}.

\section{Multiply-with-carry generators}
\label{sec:mwc}

We fix throughout the paper a \emph{base} $b\geq 2$ (in practice, $b=2^w$ with
$w=32$ or $w=64$), a \emph{lag} $k\geq 1$, a \emph{multiplier} $a>0$, and an
integer $a_0$ coprime with $b$; we let $A=a_0^{-1}\bmod b$. The generator is
defined by the recurrence
\begin{equation}
\label{eq:rec}
a_0x_n+bc_n=ax_{n-k}+c_{n-1},\qquad 0\leq x_n<b,\qquad n\geq k,
\end{equation}
where the initial state is given by \emph{digits} $x_0$, $x_1$, $\dots\,$,
$x_{k-1}\in[0\..b)$ and by a \emph{carry} $c_{k-1}\in\Z$. Since $a_0$ is
invertible modulo $b$, reading~(\ref{eq:rec}) modulo $b$ shows that $x_n$ is
uniquely determined as $x_n=A(ax_{n-k}+c_{n-1})\bmod b$, and then $c_n$ is
uniquely determined, too. The state of the generator at time $n\geq k$ is
$(x_{n-k},\dots,x_{n-1};c_{n-1})$, and the generator outputs the sequence of
digits $x_0$, $x_1$, $x_2$, \dots\,.\footnote{Implementations may return
the oldest or the newest digit of the state, but the output sequence is in
both cases the sequence of digits, possibly shifted by $k-1$ positions.}
The recurrence~(\ref{eq:rec}) is the special case of the general recurrence
studied by Goresky and Klapper~\cite{GoKEMCRNGMP} in which only the coefficients
of $x_{n-k}$ and $x_n$ are nonzero. The fundamental parameter associated with
the generator is the \emph{modulus}, which Goresky and Klapper call
\emph{connection integer}:
\[
m = ab^k-a_0.
\]
Note that $m\equiv -a_0\pmod b$, so $m$ is coprime with $b$. We will always
assume $m>1$.

When $a_0=1$, the recurrence~(\ref{eq:rec}) becomes
\[
x_n = (ax_{n-k}+c_{n-1})\bmod b,\qquad c_n = \left\lfloor
(ax_{n-k}+c_{n-1})/b\right\rfloor,
\]
which is an MWC generator~\cite{MarYAR,CoLDPMCRNG} with a single nonzero
coefficient (as in~\cite{MarRNG}) and modulus $m=ab^k-1$. When $b=2^w$ and
$a<b$, the computation of $ax_{n-k}+c_{n-1}$ requires a $w\times w\to 2w$-bit
multiplication and a $2w$-bit addition; $x_n$ and $c_n$ are then just the lower
and upper $w$ bits of the result.

Complementary MWC generators~\cite{CoLLRCURNG,CoLDPMCRNG,MarRNG} correspond
instead to the case $a_0=-1$: for example, Marsaglia's complementary generator
computes $x_n=(b-1)-((ax_{n-k}+c_{n-1})\bmod b)$ and
$c_n=\lfloor(ax_{n-k}+c_{n-1})/b\rfloor$, and a direct computation shows that
$-x_n+(c_n+1)b=ax_{n-k}+(c_{n-1}+1)$, that is, the recurrence~(\ref{eq:rec})
with $a_0=-1$ and carry $c_n+1$. Their modulus is $ab^k+1$.

In this paper, a GMWC generator will be a generator with $a_0=-e$, where $e>0$
(odd, as $b$ will be a power of two); complementary MWC generators are thus the
case $e=1$, and all our results about GMWC generators apply to them. In
particular, since $|a_0|=1$, complementary MWC generators have the same
structural limitations as MWC generators (Section~\ref{sec:lacunary}), and when
$b=2^w$ their modulus cannot be a safe prime, as $(m-1)/2$ is even. We will be
interested, instead, in large values of $e$. We use the name
GMWC for the family of Goresky and Klapper, who simply call these generators
MWC generators. In this case
\[
x_n = A(ax_{n-k}+c_{n-1})\bmod b,\qquad c_n = (ax_{n-k}+c_{n-1}+ex_n)/b,
\]
where the division is exact, and $m=ab^k+e$. The computation now requires an additional $w\times w\to w$-bit
multiplication (by $A$), an additional $w\times w\to 2w$-bit multiplication (by
$e$), and an additional $2w$-bit addition. Note that both multiplications lie on
the dependency chain of the carry (see Section~\ref{sec:speed}). We restrict
ourselves to negative $a_0$ because in that case all carries of strictly
periodic states are nonnegative (see Section~\ref{sec:seeding}), which makes it
possible to implement the generator using unsigned arithmetic.

The formula for the carry given by Goresky and Klapper~\cite{GoKEMCRNGMP} is
not correct. In their notation, the state is $(x_{-1},\dots,x_{-r};c)$, where
$x_{-1}$ is the most recent digit, and the next digit $x'_{-1}$ and carry $c'$
satisfy $a_0x'_{-1}+c'b=\tau$, where $\tau=\sum_{i=1}^r a_ix_{-i}+c$ (in our
case, this is~(\ref{eq:rec})). They then
state~\cite[Equations~(3) and~(4)]{GoKEMCRNGMP} that $x'_{-1}=A\tau\bmod b$ and
$c'=(\tau-a_0x'_{-1})/b=\tau\operatorname{div}b$. The first expression for $c'$
is the definition solved for $c'$, but the second one is false in general: since
$\tau\operatorname{div}b=(\tau-(\tau\bmod b))/b$, the two expressions differ by
$((\tau\bmod b)-a_0x'_{-1})/b$, which is zero if and only if $0\leq
a_0x'_{-1}<b$. This is always true when $a_0=1$, as in Marsaglia's MWC
generators, but if $a_0<0$, as in our GMWC generators, the correct carry is
strictly larger than $\tau\operatorname{div}b$ whenever
$x'_{-1}\neq0$.\footnote{The error has been confirmed by Mark Goresky (personal
communication, October 2020).} The same formula appears in the later
monograph~\cite[Chapter~4]{GoKASRS}, whereas the definition of algebraic
feedback shift registers of Klapper and Xu~\cite{KlXAFSR} computes the carry
correctly. The proofs in~\cite{GoKEMCRNGMP} use the correct expression (e.g.,
Equation~(9) in the proof of their Theorem~3.1): the error affects only
implementations following Equation~(4) literally.

We will name generators by their state size: for $b=2^{64}$, \mwc[128] has
$k=1$, \mwc[192] has $k=2$, and \mwc[256] has $k=3$, and similarly for GMWC
generators. Generators with $b=2^{32}$ will be denoted by a \texttt{/32} suffix
(e.g., \mwc[64/32] has $b=2^{32}$ and $k=1$). Figure~\ref{fig:code} shows C
implementations of \mwc[128], \gmwc[128], \mwc[192], and \gmwc[256] with the
scrambler X of Section~\ref{sec:x}, and Figure~\ref{fig:code32} of \mwc[64/32]
and \mwc[128/32].

\begin{figure}
\small
\begin{verbatim}
    #define MWC_A1 0xff6da236274e252e
    /* Seed: 0 < mwc_c < MWC_A1 - 1 (we suggest mwc_c = 1) */
    static uint64_t mwc_x, mwc_c;
    static inline uint64_t mwc128x_next(void) {
        const uint64_t result = mwc_x ^ mwc_c; /* Scrambler X */
        const __uint128_t t = MWC_A1 * (__uint128_t)mwc_x + mwc_c;
        mwc_x = t; mwc_c = t >> 64;
        return result;
    }
\end{verbatim}
\medskip
\begin{verbatim}
    #define GMWC_A1 0xfff0cfb4ed6366fa
    #define GMWC_MINUSA0 0x68935f57ULL    /* e */
    #define GMWC_A0INV 0x8316d3ad5de20399 /* (-e)^-1 mod 2^64 */
    /* Seed: GMWC_MINUSA0 <= gmwc_c <= GMWC_A1 (a sufficient condition,
       not an invariant; we suggest gmwc_c = 2 * GMWC_MINUSA0) */
    static uint64_t gmwc_x, gmwc_c;
    static inline uint64_t gmwc128x_next(void) {
        const __uint128_t t = GMWC_A1 * (__uint128_t)gmwc_x + gmwc_c;
        gmwc_x = GMWC_A0INV * (uint64_t)t;
        gmwc_c = (t + GMWC_MINUSA0 * (__uint128_t)gmwc_x) >> 64;
        return gmwc_x ^ (uint64_t)(t >> 64); /* Scrambler X */
    }
\end{verbatim}
\medskip
\begin{verbatim}
    #define MWC_A2 0xff4e24dd418f5433
    /* Seed: 0 < mwc_c < MWC_A2 - 1 (we suggest mwc_c = 1) */
    static uint64_t mwc_x, mwc_y, mwc_c;
    static inline uint64_t mwc192x_next(void) {
        const uint64_t result = mwc_y ^ mwc_c; /* Scrambler X */
        const __uint128_t t = MWC_A2 * (__uint128_t)mwc_x + mwc_c;
        mwc_x = mwc_y; mwc_y = t; mwc_c = t >> 64;
        return result;
    }
\end{verbatim}
\medskip
\begin{verbatim}
    #define GMWC_A3 0xff76e600f25ae046
    #define GMWC_MINUSA0 0x34dabaf938e45b /* e */
    #define GMWC_A0INV 0xac8f71df86d6542d /* (-e)^-1 mod 2^64 */
    /* Seed: GMWC_MINUSA0 <= gmwc_c <= GMWC_A3 (a sufficient condition,
       not an invariant; we suggest gmwc_c = 2 * GMWC_MINUSA0) */
    static uint64_t gmwc_x, gmwc_y, gmwc_z, gmwc_c;
    static inline uint64_t gmwc256x_next(void) {
        const __uint128_t t = GMWC_A3 * (__uint128_t)gmwc_x + gmwc_c;
        gmwc_x = gmwc_y; gmwc_y = gmwc_z; gmwc_z = GMWC_A0INV * (uint64_t)t;
        gmwc_c = (t + GMWC_MINUSA0 * (__uint128_t)gmwc_z) >> 64;
        return gmwc_z ^ (uint64_t)(t >> 64); /* Scrambler X */
    }
\end{verbatim}
\Description{C source code of the next-state functions of MWC128, GMWC128, MWC192, and GMWC256 with the scrambler X.}
\caption{\label{fig:code}The C code for \mwc[128], \gmwc[128], \mwc[192], and
\gmwc[256] with the scrambler X of Section~\ref{sec:x}; returning
\texttt{mwc\_x}, \texttt{gmwc\_x}, \texttt{mwc\_y}, and \texttt{gmwc\_z},
respectively, gives the plain generators. The type \texttt{\_\_uint128\_t} is a GCC/Clang
extension; other compilers provide equivalent intrinsics (e.g.,
\texttt{\_umul128}). Seeding is discussed in Section~\ref{sec:seeding}.}
\end{figure}

\begin{figure}
\small
\begin{verbatim}
    #define MWC_A1 0xfff145f8
    /* Seed: 0 < mwc_c < MWC_A1 - 1 (we suggest mwc_c = 1) */
    static uint32_t mwc_x, mwc_c;
    static inline uint32_t mwc64x32_next(void) {
        const uint32_t result = mwc_x ^ mwc_c; /* Scrambler X */
        const uint64_t t = MWC_A1 * (uint64_t)mwc_x + mwc_c;
        mwc_x = t; mwc_c = t >> 32;
        return result;
    }
\end{verbatim}
\medskip
\begin{verbatim}
    #define MWC_A3 0xffcccc66
    /* Seed: 0 < mwc_c < MWC_A3 - 1 (we suggest mwc_c = 1) */
    static uint32_t mwc_x, mwc_y, mwc_z, mwc_c;
    static inline uint32_t mwc128x32_next(void) {
        const uint32_t result = mwc_z ^ mwc_c; /* Scrambler X */
        const uint64_t t = MWC_A3 * (uint64_t)mwc_x + mwc_c;
        mwc_x = mwc_y; mwc_y = mwc_z; mwc_z = t; mwc_c = t >> 32;
        return result;
    }
\end{verbatim}
\Description{C source code of the next-state functions of MWC64/32 and MWC128/32 with the scrambler X.}
\caption{\label{fig:code32}The C code for \mwc[64/32] and \mwc[128/32] with the
scrambler X of Section~\ref{sec:x}; returning \texttt{mwc\_x} and
\texttt{mwc\_z}, respectively, gives the plain generators.}
\end{figure}

\section{The arithmetic of (G)MWC generators}
\label{sec:arith}

We now recall, with complete proofs, the basic properties of (G)MWC
generators, following the $b$-adic approach of Goresky and
Klapper~\cite{GoKEMCRNGMP}, which gives very short proofs.

\subsection{\texorpdfstring{$b$}{b}-adic numbers}

We recall a few basic facts about $b$-adic numbers; we refer the reader to the
monograph by Goresky and Klapper~\cite{GoKASRS}, where they are called
$N$-adic numbers, for proofs and details. A $b$-adic integer is a formal series
$\alpha=\sum_{i\geq0}x_ib^i$ with $0\leq x_i<b$; we call $x_0$, $x_1$, \dots\
its \emph{digits}. The $b$-adic integers form a ring $\Zb$, in which sums and
products are computed propagating carries, exactly like in base-$b$ arithmetic,
except that numbers are infinite to the left. 

The ring $\Zb$ contains $\Z$ (e.g., $-1=\sum_{i\geq0}(b-1)b^i$),
and, more generally, all rationals whose denominator is coprime with $b$; the
map from such rationals to $\Zb$ is injective. A series of $b$-adic integers
$\sum_n \beta_n$ converges if the number of trailing zero digits of $\beta_n$
goes to infinity; in particular, $\sum_{n}c_nb^n$ converges for every sequence
of integers $c_n$. When $b=2^w$, $\Zb$ is the ring $\Z_2$ of $2$-adic
integers, and the base-$b$ digits of $\alpha\in\Z_2$ are its bits grouped in
blocks of $w$.

The digits of a rational number in $\Zb$ are eventually periodic, and the
fundamental fact we will use is the following characterization of
\emph{strictly} periodic expansions, which is the $b$-adic analogue
of~\cite[Theorem~2.1]{KlGFSRSCM}:

\begin{lemma}
\label{lemma:periodic}
Let $m>1$ be coprime with $b$ and $h\in\Z$. Then, the $b$-adic expansion of
$-h/m$ is strictly periodic if and only if $0\leq h\leq m$.
\end{lemma}
\begin{proof}
Let $P=\ord_m(b)$ and $b^P-1=qm$. If $0\leq h\leq m$, then
$-h/m=hq/(1-b^P)=hq\sum_{j\geq0}b^{jP}$; since $0\leq hq\leq b^P-1$, the
expansion of $-h/m$ is the $P$-digit base-$b$ representation of $hq$ repeated
forever. Conversely, if the expansion is strictly periodic with some period $p$,
and $D\in[0\..b^p)$ is the integer represented by the first $p$ digits, then
$-h/m=D/(1-b^p)$, that is, $h/m=D/(b^p-1)\in[0\..1]$.
\end{proof}

The proof shows also that the $b$-adic digits of $-h/m$ are the digits of the
ordinary base-$b$ fractional expansion of $h/m\in[0\..1)$ in reverse order,
which is why the classical analyses~\cite{TLCLSACSBRNG,CoLDPMCRNG}, based on
fractional expansions, describe carry generators as producing the digits of a
rational number ``in reverse order'' (see also~\cite{MarRNG}).

\subsection{The state integer}

The fundamental observation is that the output of a (G)MWC generator is the
$b$-adic expansion of a rational number with denominator $m$, whose numerator
can be computed from the state:

\begin{lemma}[\cite{GoKEMCRNGMP}, Lemma~6.1]
\label{lemma:state}
Let $(x_0,\dots,x_{k-1};c_{k-1})$ be a state of the generator, let
$X=\sum_{j=0}^{k-1}x_jb^j$ be the integer whose base-$b$ digits are $x_{k-1}$,
\dots, $x_0$, and let
\begin{equation}
\label{eq:h}
h=c_{k-1}b^k+a_0X.
\end{equation}
Then, $\sum_{i\geq0}x_ib^i=-h/m$ in $\Zb$.
\end{lemma}
\begin{proof}
Let $\alpha=\sum_{i\geq0}x_ib^i$. Multiplying~(\ref{eq:rec}) by $b^n$ and
summing over $n\geq k$ we obtain
\[
a_0(\alpha-X)+\sum_{n\geq k}c_nb^{n+1} = ab^k\alpha+c_{k-1}b^k+\sum_{n\geq
k}c_nb^{n+1},
\]
whence $(a_0-ab^k)\alpha=a_0X+c_{k-1}b^k$, that is, $-m\alpha=h$.
\end{proof}

We call $h$ the \emph{state integer}. Note that for an MWC generator ($a_0=1$)
the state integer is simply the number with base-$b$ digits $c_{k-1}$,
$x_{k-1}$, \dots, $x_0$, that is, the state read as a $(k+1)$-digit number with
the carry as most significant digit. For a GMWC generator, $h=c_{k-1}b^k-eX$.
As noted by Goresky and Klapper, the state integer determines the state:
reading~(\ref{eq:h}) modulo $b$, $b^2$, \dots, $b^k$ we recover successively
$x_0$, $x_1$, \dots, $x_{k-1}$, and then $c_{k-1}$.

\subsection{The associated MCG}

We say that a state is \emph{strictly periodic} if the generator returns to it,
and that it is \emph{degenerate} if the output is constant. Since the state
integer determines the state, and the map from rationals to $\Zb$ is injective,
the generator returns to a state if and only if the output sequence from that
state is strictly periodic; thus, by Lemmas~\ref{lemma:periodic}
and~\ref{lemma:state}, a state is strictly periodic if and only if $0\leq h\leq
m$. The all-zero state ($h=0$) is degenerate. The state with all digits equal to
$b-1$ and carry $a-a_0$ (the value of the carry is computed from~(\ref{eq:rec})
with $x_n=x_{n-k}=b-1$) has $h=m$, because $-1=\sum_{i\geq0}(b-1)b^i$, and is
degenerate, too.

\begin{theorem}
\label{teo:mcg}
Let $B=b^{-1}\bmod m$, consider a generator started from a state with
$0<h_0<m$, and for $i\geq0$ let $h_i$ be the state integer of the state
$(x_i,\dots,x_{i+k-1};c_{i+k-1})$. Then, for all $i\geq0$:
\begin{enumerate}
\item $0<h_i<m$ and $h_{i+1}=Bh_i\bmod m$;
\item $x_i=Ah_i\bmod b$;
\item $x_i=\lfloor bh_{i+1}/m\rfloor$.
\end{enumerate}
\end{theorem}
\begin{proof}
It is sufficient to prove the statements for $i=0$. Since the next state
outputs $x_1$, $x_2$, \dots, by Lemma~\ref{lemma:state} we have
$-h_0/m=x_0+b(-h_1/m)$, that is,
\begin{equation}
\label{eq:step}
bh_1=h_0+mx_0,
\end{equation}
an equation among integers. Reading it modulo $m$ gives $h_1\equiv Bh_0\pmod
m$; reading it modulo $b$, and using $m\equiv -a_0\pmod b$, gives $a_0x_0\equiv
h_0 \pmod b$, that is, $x_0=Ah_0\bmod b$. The initial state is strictly
periodic, and so is the next state, so $0\leq h_1\leq m$; if $h_1=0$,
then~(\ref{eq:step}) implies $h_0\leq 0$, and if $h_1=m$ then~(\ref{eq:step})
implies $h_0=m(b-x_0)\geq m$, so $0<h_1<m$ and $h_1=Bh_0\bmod m$. Finally,
$bh_1/m-x_0=h_0/m\in[0\..1)$, so $x_0=\lfloor bh_1/m\rfloor$.
\end{proof}

Thus, the state integers are generated by an MCG with modulus $m$ and
multiplier $B=b^{-1}\bmod m$, the analogue of the \emph{LCG representation} of
AWC and SWB generators~\cite{TLCLSACSBRNG}; Theorem~\ref{teo:mcg} corresponds
to Theorem~1 of~\cite{CoLDPMCRNG} in the MWC case, and to Theorems~2.2 and~4.1
of~\cite{GoKEMCRNGMP} in general (the sign of $A$
in~\cite[Theorem~2.2]{GoKEMCRNGMP} is incorrect, see
Section~\ref{sec:erratum}). For MWC generators $A=1$, and when $k=1$ we have
$ab\equiv 1\pmod m$, so $B=a$: the state integers of a lag-one MWC generator are
generated by the MCG with modulus $ab-1$ and multiplier $a$~\cite{MarRNG}.

\subsection{Connections with algebraic recurrences}

The output of a linear feedback shift register (LFSR) over $\F_2$ with
connection polynomial $Q(x)$ is the power-series expansion of a rational
function $P(x)/Q(x)$, whereas the output of a (G)MWC generator is the $b$-adic
expansion of a rational number $-h/m$. When $b=2^w$, since $\Zb=\Z_2$, a (G)MWC
generator with modulus $m$ computes, $w$ bits at a time, the bit sequence of a
binary FCSR with connection integer $m$~\cite{KlGTASR,KlGFSRSCM}, in the same
way as word-oriented linear generators compute many steps of an LFSR at once.
Algebraic feedback shift registers~\cite{KlXAFSR} provide a common
generalization (see~\cite{GoKASRS} for a comprehensive treatment), and the
second statement of Theorem~\ref{teo:mcg} is the analogue of the trace
representation of $m$-sequences, called \emph{exponential
representation}~\cite{KlGFSRSCM,KlXAFSR}.

\subsection{Period}

By Theorem~\ref{teo:mcg}, if $m$ is prime the period of the output sequence of
a generator started from a strictly periodic nondegenerate state is
$\ord_m(b)$. The period is $m-1$, the maximum possible, if and only if $b$ is a
primitive root modulo $m$. As observed by Couture and
L'Ecuyer~\cite[Section~3]{CoLDPMCRNG} (see also~\cite{GoKEMCRNGMP}), when
$a_0=1$ and $b=2^w$ with $w\geq3$ this never happens: indeed, $m\equiv -1
\pmod 8$, so $2$ is a quadratic residue modulo $m$ (and so is $b$). Moreover,
if $w$ is even $b$ is a square, so for $b=2^{32}$ or $b=2^{64}$ the period of
any (G)MWC generator is at most $(m-1)/2$.

This bound is easily attained: if $m$ is a \emph{safe prime}, that is, both $m$
and $(m-1)/2$ are prime, then $\ord_m(b)$ is a divisor of $m-1$; since $b$ is a
quadratic residue, $\ord_m(b)$ divides the prime $(m-1)/2$, and since $1<b<m$
we conclude that $\ord_m(b)=(m-1)/2$. In this case, the nonzero residues modulo
$m$ are split into two orbits of size $(m-1)/2$ (the quadratic residues and the
quadratic nonresidues), and every nondegenerate strictly periodic state has
period $(m-1)/2$. Safe primes were already used for this purpose by Marsaglia,
as reported by Couture and L'Ecuyer~\cite[p.~605]{CoLDPMCRNG}. All the (G)MWC
generators of our search have a safe-prime modulus.

If $m>5$ is a safe prime, then $(m-1)/2$ is odd, so $m\equiv 3\pmod 4$ and $-1$
is a quadratic nonresidue: the map $h\mapsto m-h$ exchanges the two orbits, and
since $-(m-h)/m=-1-(-h/m)$ and $-1=\sum_{i\geq0}(b-1)b^i$, the outputs of the
two orbits are digit-wise complements of each other. In every dimension, thus,
the points of the two orbits lie in the same families of parallel hyperplanes,
and, as is customary when the generator has several
cycles~\cite[Section~5.3]{TLCLSACSBRNG}, we will study the lattice structure of
the points produced by all cycles together.\footnote{Over a single orbit of an
MWC generator, the excess of zeroes over ones in each output bit is the class
number of $\Q(\sqrt{-m})$~\cite[Section~3]{CoLDPMCRNG}, a negligible quantity
for our moduli.}

\subsection{Errata}
\label{sec:erratum}

Besides the formula for the carry discussed in Section~\ref{sec:mwc},
\cite{GoKEMCRNGMP} contains two other errors. The first one is a sign error in
the exponential representation of the output. Theorem~2.2
of~\cite{GoKEMCRNGMP} states that the output of a strictly periodic state with
state integer $h$ (defined as in Lemma~\ref{lemma:state}) is
\[
x_{-r+i}=-A\bigl(hB^i\bmod m\bigr)\bmod b,
\]
where $A=a_0^{-1}\bmod b$ and $B=b^{-1}\bmod m$, but by
Theorem~\ref{teo:mcg} the correct sign is positive. The problem is in the
proof: reading the identity $h=bh'-mx_{-r}$ modulo $b$, Goresky and Klapper use
$m\equiv a_0\pmod b$, whereas their definition of the connection integer,
$m=-a_0+\sum_{i=1}^ra_ib^i$, gives $m\equiv-a_0\pmod b$. The error is visible
already for MWC generators: for $b=10$, $a_0=1$, $a=2$, and $k=1$ we have $m=19$,
and the state $(1;0)$ has state integer $h=1$ and output $1$, $2$, $4$, $8$,
$6$, \dots, whereas the formula gives $9$ as first output. The same sign error
occurs in the proof of Theorem~4.1 of~\cite{GoKEMCRNGMP} (when reading their
Equation~(10) modulo $b$), so the map $\varphi(h)=-Ah\bmod b$ used in that
theorem, and defined in~\cite[Section~2]{GoKEMCRNGMP}, should be
$\varphi(h)=Ah\bmod b$.

In the later monograph~\cite[Chapter~4]{GoKASRS} the output is written as the
expansion of $f/m$, where $f=-h$, but the representation is stated as
$A(B^if\bmod m)\bmod b$, with the reduction modulo $m$ in $[0\..m)$: since
$Am\equiv-1\pmod b$, this formula yields the digit-wise complement $b-1-x_i$ of
the actual output. Both incorrect formulas replace each output digit $x_i$ by a fixed
bijection of it ($-x_i\bmod b$ in~\cite{GoKEMCRNGMP}, $b-1-x_i$
in~\cite{GoKASRS}), which preserves the period and the distribution of tuples of
outputs: thus, the other results of~\cite{GoKEMCRNGMP} are not affected.

The second error concerns the table of parameters. Goresky and
Klapper~\cite[Section~7]{GoKEMCRNGMP} correctly observe that $b=2^w$
is a primitive root modulo a prime $m$ if and only if $2$ is a primitive root
modulo $m$ and $w$ is coprime with $m-1$, and then propose a table of twelve
parameter sets, with $w\in\{21,23,24,25,31,32,33,35\}$, claimed to yield period
$m-1$. However, three of the parameter sets have an even $w$, which cannot be
coprime with $m-1$. We verified using Sage that for all twelve parameter sets
$m$ is a (probable) prime and $2$ is a primitive root modulo $m$. Thus, the nine parameter sets with odd $w$ are correct, but for the
two parameter sets with $w=24$ the period is $(m-1)/2$, and for the parameter
set with $w=32$ (where $\gcd(32,m-1)=4$) the period is $(m-1)/4$ (see
Table~\ref{tab:erratum}). Mark Goresky confirmed that the three parameter sets
should not have been included.\footnote{Personal communication, October 2020.}
The claims about the distribution of tuples of outputs
in~\cite[Theorem~2.4]{GoKEMCRNGMP} do not apply to these three parameter sets.

\begin{table}
\caption{\label{tab:erratum}The parameter sets of~\cite[Section~7]{GoKEMCRNGMP}
that do not have maximum period. The values of $p$ and $m$ are as
in~\cite{GoKEMCRNGMP}.}
\renewcommand{\arraystretch}{1.1}
\begin{tabular}{lllcc}
\multicolumn{1}{c}{$b$} & \multicolumn{1}{c}{$p$} & \multicolumn{1}{c}{$m$} & Claimed period & Actual period\\
\hline
$2^{24}$ & $b^{48}-b^{46}-b^{38}-b^{14}+1$ & $2p+1$ & $m-1$ & $(m-1)/2$\\
$2^{24}$ & $b^{41}-b^{38}-2b^{14}+1$ & $2p+1$ & $m-1$ & $(m-1)/2$\\
$2^{32}$ & $b^{33}-b^{20}-b^{14}-b^{11}-b^4+1$ & $4p+1$ & $m-1$ & $(m-1)/4$\\
\end{tabular}
\end{table}

\subsection{Seeding}
\label{sec:seeding}

A user must be able to seed a generator easily, avoiding degenerate states and
states that are not strictly periodic. For an MWC generator this is simple: since the state integer is the state
read as a number, $0\leq h\leq m=ab^k-1$ if and only if $0\leq c_{k-1}\leq
a-1$, and, if $m$ is a prime larger than $b$ (as in all our generators), the
only degenerate states in this range are the all-zero state and the state with
all digits equal to $b-1$ and carry $a-1$.\footnote{A degenerate state with
constant output $x$ has $-h/m=x/(1-b)$, that is, $h(b-1)=xm$; if $m$ is coprime
with $b-1$ this implies $m\mid h$, so $h=0$ or $h=m$. Without this hypothesis
there may be other degenerate states.} Thus, it is sufficient to choose the
digits arbitrarily and the carry so that $0\leq c_{k-1}\leq a-2$, excluding the
all-zero state (this condition is slightly conservative, as it excludes all
states with carry $a-1$). These facts are classical: the range of the carry of
strictly periodic states is given by Couture and
L'Ecuyer~\cite[Corollary~3]{CoLDPMCRNG}.

For GMWC generators the situation is more complex, as the state integer is not
the state read as a number. The degenerate states are again the all-zero state
and the state with all digits equal to $b-1$ and carry $a+e$, but states lying
``between'' these two states (as numbers) are not necessarily strictly periodic:
for example, if $c_{k-1}=0$ and $x_0\neq 0$ then $h<0$. Theorem~2.1
of~\cite{GoKEMCRNGMP} characterizes strictly periodic states in terms of the
state integer (as we did above), but checking whether the state integer of
a state is in $[0\..m]$ is cumbersome. The following sufficient condition is
much more practical:

\begin{proposition}
\label{prop:seed}
Consider a GMWC generator with $a_0=-e<0$. Then, every state whose carry
satisfies $e\leq c_{k-1}\leq a$ has state integer $0<h<m$.
In particular, if $m$ is prime, all such states are strictly periodic,
nondegenerate, and have period $\ord_m(b)$.
\end{proposition}
\begin{proof}
Since $X\in[0\..b^k)$, we have $h=c_{k-1}b^k-eX\geq e(b^k-X)\geq e>0$, and
$h\leq ab^k-eX\leq ab^k<ab^k+e=m$. The second statement follows from
Theorem~\ref{teo:mcg}.
\end{proof}

The condition of Proposition~\ref{prop:seed} excludes some strictly periodic
states (e.g., if $e>1$, the state with $c_{k-1}=1$, $x_0=1$, and all other
digits equal to zero, for which $h=b^k-e$), but in practice the loss is
negligible: there are $(a-e+1)b^k$ states satisfying the condition, out of the
$m-1=ab^k+e-1$ nondegenerate strictly periodic states. If $a\geq b-b/2^8$ (as
it will happen for all our multipliers) and $a+e<b$ (see below), then
$e/a<1/255$, so we lose less than $0.4\%$ of the states.\footnote{The condition of Proposition~\ref{prop:seed} is a condition for
\emph{seeding}, not an invariant: strictly periodic states may have carries
outside $[e\..a]$, and a routine validating a state (e.g., one restored from a
file) should rather check that $0<h<m$.}

As for the range of the carry, if $0\leq h\leq m$ then
$c_{k-1}=(h+eX)/b^k\in[0\..a+e]$.\footnote{This result is also contained in
Theorem~3.1 of~\cite{GoKEMCRNGMP} and in~\cite[Chapter~4]{GoKASRS}; for an MWC
generator, the range is $[0\..a)$.} Thus, if $a+e<b$, the carry can be stored in
a single $w$-bit word, and the intermediate results
$ax_{n-k}+c_{n-1}<ab+b\leq b^2$ and $ax_{n-k}+c_{n-1}+ex_n<(a+e+1)b\leq b^2$
fit into $2w$ bits, as assumed by the code in Figures~\ref{fig:code}
and~\ref{fig:code32}.

Finally, since the state integers evolve as an MCG, states whose state integers
are in a simple arithmetic relation produce outputs in the same relation. In
particular, the carry of an MWC generator should never be set to zero: in that case the state
integer is $X$, so, for example, the outputs of \mwc[128] seeded with $x=2$ and
$x=3$ are, modulo $2^{64}$ and up to a difference of at most two, twice and
three times the outputs of the generator seeded with $x=1$. 
We thus suggest to
set the carry to one: the state integer $b^k+X$ then lies in $[b^k\..2b^k)$, so
no state integer is a small multiple of another.

The same problem
occurs for a GMWC generator with carry $a$, whose state integer is then
congruent to $-e(X+1)$ modulo $m$, and with carry $e$, whose state integer is
$e(b^k-X)$; the carry $2e$, which places the state integer $e(2b^k-X)$ in
$(eb^k\..2eb^k]$, avoids it.

In all cases, however, generators seeded with consecutive values of $X$ and the
same carry have state integers in arithmetic progression, so the second
difference of the outputs of three such generators is, modulo $b$, always in
$\{-1,0,1\}$: the state should be derived from the seed using a good mixing
function (e.g., SplitMix~\cite{SLFFSPNG}), mapping the carry into the allowed
range, or parallel streams should be obtained by jumping.

\subsection{Jumping ahead}
\label{sec:jump}

Theorem~\ref{teo:mcg} provides an immediate way to jump ahead, as already
observed for AWC and SWB generators by Tezuka, L'Ecuyer, and
Couture~\cite[Section~3]{TLCLSACSBRNG}: to advance a generator by $J$ steps, one
computes the state integer $h$, multiplies it by the precomputed constant
$B^J\bmod m$, and recovers the state from the result. For MWC generators the
first and last step are trivial, as the state integer is just the state read as
a number. For GMWC generators we have $h=c_{k-1}b^k-eX$, and given $h$ we recover
the digits by $X=(-e)^{-1}h\bmod b^k$ (for $k=1$, $x_0=Ah\bmod b$) and the carry
by $c_{k-1}=(h+eX)/b^k$. The only cost is a modular multiplication on
$(k+1)w$-bit integers; jumping backwards is equally easy, as $B^{-1}=b$. With a
jump of $J$ steps one obtains $\lfloor (m-1)/(2J)\rfloor$ non-overlapping streams
(e.g., about $2^{63}$ streams of length $2^{64}$ for \mwc[128]).

\ifanon{Our code}{The code distributed by the author}, which uses the multipliers suggested in
Section~\ref{sec:results}, can jump by $J=2^{(k+1)w/2}$, by $J=2^{3(k+1)w/4}$
(where $b=2^w$), and by an arbitrary number of steps (by modular
exponentiation).

\section{The lattice structure of (G)MWC generators}
\label{sec:spectral}

\subsection{The spectral test}
\label{sec:spectraltest}

Consider an MCG with modulus $m$ and multiplier $\beta$, and let
\[
\Psi_t=\bigl\{\,(h,\beta h\bmod m,\dots,\beta^{t-1}h\bmod m)/m\bigm| 0\leq h<m\,\bigr\}
\]
be the set of all vectors of $t$ successive values produced by the generator
from all initial states (including the zero vector). Then, $\Psi_t=L_t\cap[0\..1)^t$,
where $L_t$ is the lattice generated by $\Z^t$ and by the vector
$(1,\beta,\dots,\beta^{t-1})/m$; since $L_t$ contains $\Z^t$ with index $m$, its
determinant is $1/m$. The maximal distance $d_t$ between successive hyperplanes,
over all families of equidistant parallel hyperplanes covering $L_t$, is
$1/\nu_t$, where $\nu_t$ is the Euclidean length of a shortest nonzero vector in
the \emph{dual lattice}
\[
L^*_t(\beta)=\bigl\{\,\bm z\in\R^t\bigm| \bm z\cdot\bm v\in\Z\text{ for all $\bm v\in L_t$}\,\bigr\}
=\Bigl\{\,(z_0,\dots,z_{t-1})\in\Z^t\Bigm|\sum_{i=0}^{t-1}z_i\beta^i\equiv
0\pmod m\,\Bigr\}:
\]
indeed, every nonzero $\bm z\in L^*_t(\beta)$ defines the family of hyperplanes
$\{\,\bm v\mid\bm z\cdot\bm v=j\,\}$, $j\in\Z$, which are $1/\|\bm z\|$ apart and
cover $L_t$. Computing $\nu_t$ is the \emph{spectral
test}~\cite{CoMFAURNG,KnuACPII}.

The dual lattice $L^*_t(\beta)$ has determinant $m$. The \emph{Hermite constant}
$\gamma_t$ is the smallest constant such that every $t$-dimensional lattice with
determinant $D$ contains a nonzero vector of length at most
$\gamma_t^{1/2}D^{1/t}$; thus,
\[
\nu_t\leq\nu^*_t(m)=\gamma_t^{1/2}m^{1/t},
\]
or, equivalently, $d_t\geq d^*_t(m)=1/\nu^*_t(m)$. Following
L'Ecuyer~\cite{LEcTLCGDSGLS}, we use the normalized \emph{figure of merit}
\[
f_t=\frac{\nu_t}{\nu^*_t(m)}=\frac{d^*_t(m)}{d_t}\in(0\..1]
\]
(L'Ecuyer writes $\ell_t$ and $S_t$ for $\nu_t$ and $f_t$): values close to one denote a good lattice structure,
whereas values close to zero denote a bad one. For $2\leq t\leq 8$ the Hermite constants are known
exactly~\cite{CoSSPLG}: $\gamma_2^2=4/3$, $\gamma_3^3=2$, $\gamma_4^4=4$,
$\gamma_5^5=8$, $\gamma_6^6=64/3$, $\gamma_7^7=64$, and $\gamma_8=2$. Since the
normalization hides the actual distance between hyperplanes, we will also report
$\log_2\nu_t$: the hyperplanes are at distance $d_t=2^{-\log_2\nu_t}$.

More generally, given a set of indices $I=\{i_1,i_2,\dots,i_t\}$ with
$0=i_1<i_2<\dots<i_t$, one can apply the spectral test to the
\emph{projection}
\[
\Psi_I=\bigl\{\,(\beta^{i_1}h\bmod m,\beta^{i_2}h\bmod m,\dots,\beta^{i_t}h\bmod m)/m\bigm|0\leq h<m\,\bigr\}
\]
of $\Psi_{i_t+1}$ on the coordinates in $I$, which, if $I$ is not of the form
$\{0,1,\dots,t-1\}$, contains vectors of non-successive values, or vectors with
\emph{lacunary indices}~\cite{LEcBLSVNVPSLR,LeLVRLR}. The dual lattice
$L^*_I(\beta)$ of the projection is formed by the vectors $(z_i)_{i\in I}\in\Z^I$
such that $\sum_{i\in I}z_i\beta^i\equiv0\pmod m$, and we denote by $\nu_I$
the length of its shortest nonzero vectors, and by $d_I=1/\nu_I$ the
corresponding distance between hyperplanes; thus, $\nu_{\{0,1,\dots,t-1\}}=\nu_t$.
Note that the vectors of $L^*_t(\beta)$ whose coordinates outside $I$ are zero
are exactly the vectors of $L^*_I(\beta)$, extended with zeroes, so
$\nu_t\leq\nu_I$ whenever $I\subseteq\{0,1,\dots,t-1\}$; in particular, $\nu_t$
is nonincreasing in $t$~\cite{KnuACPII}.

By Theorem~\ref{teo:mcg}, a (G)MWC generator simulates the MCG with multiplier
$B=b^{-1}\bmod m$: the output $x_i/b$ differs from $h_{i+1}/m$ by less than
$1/b$, so vectors of $t$ successive outputs $(x_n,\dots,x_{n+t-1})/b$ are, up
to an error smaller than $1/b$ in each coordinate, points of $\Psi_t$ for the
multiplier $B$~\cite[Theorem~1]{CoLDPMCRNG}, and the same is true of any
projection. Multiplying
the congruence defining $L^*_t(B)$ by $b^{t-1}$ we see that $L^*_t(B)$ is just
$L^*_t(b)$ with coordinates in reverse order (the reverse sequence has the same
lattice structure~\cite[Section~4]{TLCLSACSBRNG}), so we can (and will) compute
figures of merit using the more convenient multiplier $b$.

The dual lattice $L^*_t(b)$ contains two kinds of short vectors. Letting
$\bm\varepsilon_i$ be the $i$-th vector of the canonical basis, $0\leq i<t$, the
vectors
\[
\bm s_i=-b\,\bm\varepsilon_i+\bm\varepsilon_{i+1},\qquad 0\leq i\leq t-2,
\]
belong to $L^*_t(b)$, because $-b\cdot b^i+b^{i+1}=0$; and, if $t\geq k+1$, the
vector
\[
\bm u=-a_0\,\bm\varepsilon_0+a\,\bm\varepsilon_k
\]
belongs to $L^*_t(b)$, because $-a_0+ab^k=m$ (and so do its shifted copies
$-a_0\,\bm\varepsilon_i+a\,\bm\varepsilon_{i+k}$, $i+k<t$). Couture and
L'Ecuyer~\cite{CoLLSCLCSRASG,CoLDPMCRNG} used these vectors to show that in
dimensions up to $k+1$ the figures of merit are essentially fixed by $b$ and
$a$. The first statement of the following proposition is Theorem~3(i)
of~\cite{CoLDPMCRNG} (see also Theorem~1 of~\cite{CoLLSCLCSRASG}):

\begin{proposition}
\label{prop:low}
Assume $b^k<m$ (which is true whenever $a_0<0$, or $a\geq 2$ and $a_0<b^k$). For
$2\leq t\leq k$, $\nu_t=\sqrt{b^2+1}$, and thus
$f_t=\sqrt{b^2+1}\big/\nu^*_t(m)$. Moreover, $\nu_{k+1}\leq\sqrt{a^2+a_0^2}$,
and thus $f_{k+1}\leq \sqrt{a^2+a_0^2}\big/\nu^*_{k+1}(m)$.
\end{proposition}
\begin{proof}
For $t\leq k$ we have $b^t<m$, that is, $b<m^{1/t}$. The proof of
Proposition~1 of~\cite{StVCESGMCPNGjr}, which uses only the definition of the
dual lattice, shows that in this case the shortest nonzero vectors of
$L^*_t(b)$ have length $\|\bm s_0\|=\sqrt{b^2+1}$. The second statement follows
from $\bm u\in L^*_{k+1}(b)$.
\end{proof}

Since $m\approx ab^k$, for $t\leq k$ the figure of merit $f_t$ is approximately
$b^{1-(k+1)/t}/\gamma_t^{1/2}$, a very small number, whereas the bound on
$f_{k+1}$ is approximately $(a/b)^{k/(k+1)}/\gamma_{k+1}^{1/2}$, and it is
attained under mild hypotheses for MWC generators~\cite[Theorem~3]{CoLDPMCRNG},
and empirically for our GMWC generators (Table~\ref{tab:ranges}; the tables of
this section are computed from the candidates of the search of
Section~\ref{sec:search}). The bad
lattice structure in dimensions $t\leq k$ is not a real problem: the distance
between hyperplanes is about $1/b$, the resolution of the output, as $k$-tuples
of successive outputs are the $k$ most significant digits of a single state
integer divided by $m$.

\begin{table}
\caption{\label{tab:ranges}Minimum and maximum value of the figures of merit
over all candidates for $b=2^{64}$. Grayed-out dimensions are those in which the
figure of merit is essentially fixed (the maximum value is less than $1.003$
times the minimum value).}
\renewcommand{\arraystretch}{1.1}
\setlength{\tabcolsep}{2.5pt}
\small
\begin{tabular}{llrrrrrrr}
Generator & & \multicolumn{1}{c}{$f_2$} & \multicolumn{1}{c}{$f_3$} & \multicolumn{1}{c}{$f_4$} & \multicolumn{1}{c}{$f_5$} & \multicolumn{1}{c}{$f_6$} & \multicolumn{1}{c}{$f_7$} & \multicolumn{1}{c}{$f_8$}\\
\hline
\mwc[128] & min & \textcolor{gray}{$0.929$} & $1.04\times 10^{-7}$ & $1.60\times 10^{-4}$ & $0.013$ & $0.044$ & $0.069$ & $0.092$\\
 & max & \textcolor{gray}{$0.931$} & $7.00\times 10^{-4}$ & $0.946$ & $0.936$ & $0.902$ & $0.876$ & $0.849$\\
\mwc[192] & min & \textcolor{gray}{$2.17\times 10^{-10}$} & \textcolor{gray}{$0.889$} & $3.72\times 10^{-9}$ & $2.79\times 10^{-6}$ & $2.24\times 10^{-4}$ & $5.11\times 10^{-3}$ & $0.052$\\
 & max & \textcolor{gray}{$2.17\times 10^{-10}$} & \textcolor{gray}{$0.891$} & $1.64\times 10^{-5}$ & $0.012$ & $0.948$ & $0.267$ & $0.850$\\
\mwc[256] & min & \textcolor{gray}{$5.04\times 10^{-20}$} & \textcolor{gray}{$3.37\times 10^{-7}$} & \textcolor{gray}{$0.838$} & $8.08\times 10^{-10}$ & $2.86\times 10^{-7}$ & $1.87\times 10^{-5}$ & $4.24\times 10^{-4}$\\
 & max & \textcolor{gray}{$5.05\times 10^{-20}$} & \textcolor{gray}{$3.38\times 10^{-7}$} & \textcolor{gray}{$0.841$} & $1.72\times 10^{-6}$ & $5.83\times 10^{-4}$ & $0.038$ & $0.864$\\
\gmwc[128] & min & \textcolor{gray}{$0.929$} & $9.14\times 10^{-4}$ & $0.011$ & $0.028$ & $0.036$ & $0.063$ & $0.072$\\
 & max & \textcolor{gray}{$0.931$} & $0.988$ & $0.950$ & $0.935$ & $0.905$ & $0.881$ & $0.847$\\
\gmwc[256] & min & \textcolor{gray}{$5.04\times 10^{-20}$} & \textcolor{gray}{$3.37\times 10^{-7}$} & \textcolor{gray}{$0.838$} & $1.55\times 10^{-5}$ & $2.13\times 10^{-3}$ & $0.103$ & $0.105$\\
 & max & \textcolor{gray}{$5.05\times 10^{-20}$} & \textcolor{gray}{$3.38\times 10^{-7}$} & \textcolor{gray}{$0.841$} & $2.44\times 10^{-3}$ & $0.046$ & $0.870$ & $0.836$\\
\end{tabular}
\end{table}

\subsection{Dimension \texorpdfstring{$k+2$}{k+2} and lacunary triples}
\label{sec:lacunary}

The real limitation of MWC generators appears in dimension $k+2$. Couture and
L'Ecuyer~\cite[Lemma~2]{CoLDPMCRNG} observed that, in dimensions larger than
$k+1$, short vectors appear when a small multiple of the shortest vector in
dimension $k+1$ is close to $b\Z^{k+1}$, and showed that the eight nonzero
coefficients of Marsaglia's generators with base $2^{16}$ can be chosen so as to
avoid the phenomenon. When only one coefficient is nonzero, however, all
coordinates of $\bm u$ but two are zero, and the phenomenon cannot be avoided.

We now make this observation explicit and quantitative. For AWC and SWB
generators, whose modulus has base-$b$ digits in $\{0,\pm1\}$, a single dual vector
with coordinates in $\{0,\pm1\}$ shows that the points lie in planes at least
$1/\sqrt3$ apart~\cite{TLCLSACSBRNG,CoLLSCLCSRASG,LEcBLSVNVPSLR}. For an MWC
generator, instead, $\bm u$ has length about $b$, and short dual vectors arise only
as combinations of $\bm u$ and $\bm s_k$, that is, as short vectors of a
two-dimensional sublattice spanned by long vectors. Bounding such vectors
uniformly in the multiplier is the technique used by L'Ecuyer and
Touzin~\cite{LETDLRNGRM} and by L'Ecuyer and Simard~\cite{LESLSSCMRRNG} to show
that, for multiple recursive generators with few nonzero coefficients, some
projections with lacunary indices have a bad lattice structure whatever the
coefficients; the order of magnitude of our bound can be derived also from
Lemma~2 of~\cite{CoLDPMCRNG} and Dirichlet's approximation theorem. To the best
of our knowledge, however, these bounds have not been stated for (G)MWC
generators, and the exact characterization of $\nu_{k+2}$ in the second statement
of Proposition~\ref{prop:twodim}, which concerns successive values and not only a
projection, the bound of Proposition~\ref{prop:threedim} for GMWC generators
with $k\geq2$, and its generalization in Section~\ref{sec:beyond} are new.\footnote{Goresky and
Klapper~\cite[Chapter~4]{GoKASRS}, for example, only remark that for large
dimensions the distribution of tuples of MWC sequences is expected to suffer
from the same shortcomings as that of AWC and SWB generators.}

In what follows, we let
\[
\bar\nu=\gamma_2^{1/2}\bigl(a_0^2(b^2+1)+a^2\bigr)^{1/4}.
\]

\begin{proposition}
\label{prop:twodim}
For every (G)MWC generator, $\nu_{k+2}\leq\bar\nu$, and thus
\[
f_{k+2}\leq\frac{\bar\nu}{\nu^*_{k+2}(m)}=\frac{\gamma_2^{1/2}\bigl(a_0^2(b^2+1)+a^2\bigr)^{1/4}}{\gamma_{k+2}^{1/2}\,m^{1/(k+2)}}.
\]
Moreover, if $|a_0|=1$ and $4\sqrt b\leq a<b$, the shortest nonzero vectors of
$L^*_{k+2}(b)$ are the shortest nonzero vectors of the sublattice spanned by
$\bm u$ and $\bm s_k$, and thus
\[
\nu_{k+2}=\min_{(x,j)\in\Z^2\setminus\{(0,0)\}}\sqrt{x^2+(ax-bj)^2+j^2}.
\]
\end{proposition}
\begin{proof}
In dimension $k+2$, consider the vectors
$\bm u=-a_0\bm\varepsilon_0+a\bm\varepsilon_k$ and
$\bm s_k=-b\bm\varepsilon_k+\bm\varepsilon_{k+1}$. We have
$\|\bm u\|^2=a_0^2+a^2$, $\|\bm s_k\|^2=b^2+1$, and $\bm u\cdot\bm s_k=-ab$, so
the two-dimensional sublattice of $L^*_{k+2}(b)$ they span has determinant $D$
with
\[
D^2=\|\bm u\|^2\|\bm s_k\|^2-(\bm u\cdot\bm s_k)^2=(a_0^2+a^2)(b^2+1)-a^2b^2=a_0^2(b^2+1)+a^2.
\]
By the definition of the Hermite constant, the sublattice contains a nonzero
vector of length at most $\gamma_2^{1/2}D^{1/2}=\bar\nu$, and so does
$L^*_{k+2}(b)$.

For the second statement, let $\bm z=(z_0,\dots,z_{k+1})$ be a shortest nonzero
vector of $L^*_{k+2}(b)$, and $n=\|\bm z\|$. Since $|a_0|=1$ and $a<b$, by the
first statement $n\leq\bar\nu<\gamma_2^{1/2}(2b^2+1)^{1/4}<1.3\sqrt b$, so
$n<a/3<b/3$; moreover, $b>16$, as $4\sqrt b<b$. Let
$U=\sum_{i<k}z_ib^i$ and $V=z_k+bz_{k+1}$. Since $U+b^kV=\sum_iz_ib^i\equiv
0\pmod m$ and $ab^k\equiv a_0\pmod m$, the integer $\Delta=aU+a_0V$ is a multiple
of $m$. Using $|z_i|\leq n$ and $b/(b-1)<4/3$, we obtain
\[
|U|\leq n\frac{b^k-1}{b-1}<\frac49b^k\qquad\text{and}\qquad|V|\leq n(b+1)<\frac5{12}ab,
\]
so $|\Delta|<(4/9+5/12)ab^k<m$, and thus $\Delta=0$. If $z_r\neq0$ for some $1\leq r<k$,
choosing the largest such $r$ we would have $|U|\geq
b^r-n(b^r-1)/(b-1)>\frac59b^r$, and thus $|V|=a|U|>\frac59ab$, a contradiction.
Hence, $U=z_0$, and $\Delta=0$ gives $z_k+bz_{k+1}=-a_0az_0$. Letting $x=-a_0z_0$
and $j=z_{k+1}$, we conclude that $\bm z=x\bm u+j\bm s_k$, whose squared length
is $x^2+(ax-bj)^2+j^2$.
\end{proof}

For MWC and complementary MWC generators $|a_0|=1$, and since $a<b$ we have
$a_0^2(b^2+1)+a^2<2b^2+1$, so $\bar\nu<\gamma_2^{1/2}(2b^2+1)^{1/4}\approx
1.28\sqrt b$. Thus, using $m\approx ab^k$, for $a$ close to $b$ we obtain
approximately
\[
f_{k+2}\lesssim \frac{2^{1/4}\gamma_2^{1/2}}{\gamma_{k+2}^{1/2}}\,b^{-\frac{k}{2(k+2)}}.
\]
For $b=2^{64}$, this ranges from about $7\times10^{-4}$ for $k=1$ to about
$1.7\times10^{-6}$ for $k=3$ (Table~\ref{tab:bounds}). Intuitively,
$\bm u$ is almost parallel to $\bm s_k$, as both are dominated by their $k$-th
coordinate, and combining them one can cancel the large coordinates.

By the second statement of Proposition~\ref{prop:twodim}, for MWC generators
$\nu_{k+2}$ depends only on $a$ and $b$, and not on the lag, and it can be
computed by reducing a two-dimensional lattice; since $\bar\nu$ is attained only
by a hexagonal lattice, the best multipliers are those for which this lattice is
almost hexagonal. For an MWC generator, the approach of Couture and
L'Ecuyer~\cite{CoLDPMCRNG}, who chose coefficients so that the multiples
$z_1\bm u$ are far from $b\Z^{k+1}$ for $0<|z_1|<3000$, amounts instead (dropping
the bound on $z_1$) to maximizing the length $\lambda$ of a shortest nonzero
vector $(-z_1,z_1a-jb)$ of the lattice generated by $(-1,a)$ and $(0,b)$, which
ignores the coordinate $j$ of $z_1\bm u+j\bm s_k$. Since $a<b$ and
$|z_1a-jb|\leq\lambda<b$ (by Lagrange's bound, $\lambda^4\leq 4b^2/3$), we have
$|j|\leq|z_1|\leq\lambda$, and thus $\lambda\leq\nu_{k+2}\leq\sqrt2\lambda$; but
the proxy is not accurate enough for selection: on a sample of $20\,000$
candidates of our \mwc[128] search, the best $\nu_3$ is $17\%$ larger than that
of the multiplier maximizing $\lambda$.

The vectors used in the proof of Proposition~\ref{prop:twodim} have a very
specific support, and this fact has important consequences:

\begin{corollary}
\label{cor:lacunary}
Let $I=\{0,1,k+1\}$. Then:
\begin{enumerate}
\item for the multiplier $B$, $\nu_I\leq\bar\nu$; thus, for every initial state
with $0<h_0<m$, the vectors $(x_n,x_{n+1},x_{n+k+1})/b$, $n\geq0$, lie, up to an
error smaller than $1/b$ in each coordinate, in a family of equidistant parallel
planes at least $1/\bar\nu$ apart;
\item for every $t\geq k+2$, $\nu_t\leq\bar\nu$, and thus
$f_t\leq\bar\nu/\nu^*_t(m)$.
\end{enumerate}
\end{corollary}
\begin{proof}
The vectors $\bm u$ and $\bm s_k$ have support $\{0,k,k+1\}$, so the proof of
Proposition~\ref{prop:twodim} shows that $\nu_J\leq\bar\nu$ for $J=\{0,k,k+1\}$
and the multiplier $b$. Reversing the coordinates in dimension $k+2$, the indices
$\{0,k,k+1\}$ for the multiplier $b$ correspond to the indices $\{0,1,k+1\}$
for the multiplier $B$ (Section~\ref{sec:spectraltest}), and the first claim
follows from Theorem~\ref{teo:mcg}. The second statement follows from
Proposition~\ref{prop:twodim}, as $\nu_t$ is nonincreasing in $t$.
\end{proof}

Under the hypotheses of the second statement of Proposition~\ref{prop:twodim},
its proof, applied to the vectors of $L^*_J(b)$ (for which $U=z_0$), shows also
that $\nu_I=\nu_J=\nu_{k+2}$, so the planes of the first statement are
exactly $d_{k+2}$ apart.

Thus, for every MWC or complementary MWC generator with $a<b$,
\emph{independently of the multiplier and of the lag}, the vectors
$(x_n,x_{n+1},x_{n+k+1})/b$ lie in a family of equidistant parallel planes at
least about $1/(1.28\sqrt b)$ apart, that is, about $2^{-32.35}$ for $b=2^{64}$
and $2^{-16.35}$ for $b=2^{32}$. A larger lag moves the defect to a projection
with lacunary indices, and makes it invisible to the figures of merit of
dimensions $3$ to $k+1$, but it does not make it smaller: the decrease of
$f_{k+2}$ with $k$ is entirely due to the normalization by $m^{1/(k+2)}$. The
phenomenon is analogous to the bad lattice structure of vectors of
non-successive values produced by recurrences with few nonzero
coefficients~\cite{LEcBLSVNVPSLR,LETDLRNGRM,LESLSSCMRRNG}, and to the defect of AWC and SWB generators in
dimensions larger than their largest lag~\cite{TLCLSACSBRNG}.

From a practical viewpoint, the distinction between successive and lacunary
indices is not the relevant one: an application using $k+1$ outputs at each
step sees the triples $(x_n,x_{n+1},x_{n+k+1})$ as naturally as triples of
successive outputs. What matters is the number of points a test needs to detect
the defect, which depends on the distance between the planes, and thus on $b$,
but not on $k$: as we show in Section~\ref{sec:detect}, for $b=2^{32}$ the defect is detected in a few
seconds, whereas for $b=2^{64}$ our test uses $3$ to $5$\,TB of output.
The scrambler X of Section~\ref{sec:x} removes, our tests, all such defects.

For GMWC generators, instead, the first coordinate of $\bm u$ is $e$, and if
$e$ is large the bound of Proposition~\ref{prop:twodim} becomes vacuous: for
$k=1$ this happens as soon as $e\geq b^{1/3}$, approximately. There is, however,
another sublattice that limits $f_{k+2}$ when $k\geq 2$:

\begin{proposition}
\label{prop:threedim}
For every (G)MWC generator with $k\geq 2$,
$\nu_{k+2}\leq\gamma_3^{1/2}\bigl((b^2+1)(a_0^2+a^2)\bigr)^{1/6}$, and thus
\[
f_{k+2}\leq\frac{\gamma_3^{1/2}\bigl((b^2+1)(a_0^2+a^2)\bigr)^{1/6}}{\gamma_{k+2}^{1/2}\,m^{1/(k+2)}}.
\]
\end{proposition}
\begin{proof}
In dimension $k+2$, consider the vectors $\bm u=-a_0\bm\varepsilon_0+a\bm\varepsilon_k$,
$\bm s_0=-b\bm\varepsilon_0+\bm\varepsilon_1$, and
$\bm s_k=-b\bm\varepsilon_k+\bm\varepsilon_{k+1}$. Since $k\geq 2$, $\bm s_0$ and
$\bm s_k$ have disjoint supports, and $\bm s_0\cdot\bm s_k=0$. Moreover,
$\bm u\cdot\bm s_0=a_0b$ and $\bm u\cdot\bm s_k=-ab$. The Gram determinant of
the three vectors is thus
\[
(a_0^2+a^2)(b^2+1)^2-(b^2+1)a_0^2b^2-(b^2+1)a^2b^2=(b^2+1)(a_0^2+a^2),
\]
and we conclude as in the previous proof, using $\gamma_3$.
\end{proof}

For GMWC generators with $e<b$ the bound on $f_{k+2}$ is approximately
$b^{-\frac{k-1}{3(k+2)}}$, up to a factor $\gamma_3^{1/2}/\gamma_{k+2}^{1/2}$,
which is much larger than the bound for MWC generators, but still very small:
for $b=2^{64}$ it is approximately $0.023$ for $k=2$ and $2.5\times10^{-3}$ for
$k=3$. For $k\geq2$, $e\approx b^{1/3}$ is also the point at which the bound of
Proposition~\ref{prop:twodim} stops being stronger than the bound of
Proposition~\ref{prop:threedim}; for MWC generators (at least when $b/2\leq
a<b$, as for all our multipliers) the bound of Proposition~\ref{prop:twodim} is
stronger. As in Corollary~\ref{cor:lacunary}, the vectors used in the proof have
support $\{0,1,k,k+1\}$ (a set invariant under reversal), so the bound applies
to $\nu_I$ for $I=\{0,1,k,k+1\}$, and to all $\nu_t$ with $t\geq k+2$: the
corresponding distance between hyperplanes is about $b^{-2/3}$, which is much
smaller than the distance $b^{-1/2}$ of MWC generators.

Table~\ref{tab:bounds} compares the bounds with the largest values of $f_{k+2}$
found by our search (see Section~\ref{sec:search}): the best multipliers are
within $1.5\%$ of the bounds, which are thus essentially tight.

\begin{table}
\caption{\label{tab:bounds}Upper bounds on $f_{k+2}$ from
Propositions~\ref{prop:twodim} (MWC) and~\ref{prop:threedim} (GMWC with
$k\geq2$), maximized over the range of multipliers we consider, versus the
maximum value found in our search, and relative gap (bound divided by maximum,
minus one).}
\renewcommand{\arraystretch}{1.1}
\begin{tabular}{llrrr}
Generator & & \multicolumn{1}{c}{Best found} & \multicolumn{1}{c}{Bound} & \multicolumn{1}{c}{Gap}\\
\hline
\mwc[128] & $f_{3}$ & $7.00\times 10^{-4}$ & $7.01\times 10^{-4}$ & $0.020\%$\\
\mwc[192] & $f_{4}$ & $1.64\times 10^{-5}$ & $1.64\times 10^{-5}$ & $0.002\%$\\
\mwc[256] & $f_{5}$ & $1.72\times 10^{-6}$ & $1.72\times 10^{-6}$ & $0.020\%$\\
\gmwc[256] & $f_{5}$ & $2.44\times 10^{-3}$ & $2.46\times 10^{-3}$ & $1.03\%$\\
\mwc[64/32] & $f_{3}$ & $0.0282$ & $0.0282$ & $0.22\%$\\
\mwc[128/32] & $f_{5}$ & $1.33\times 10^{-3}$ & $1.34\times 10^{-3}$ & $0.44\%$\\
\gmwc[128/32] & $f_{5}$ & $0.0468$ & $0.0474$ & $1.24\%$\\
\end{tabular}
\end{table}

\subsection{Beyond dimension \texorpdfstring{$k+2$}{k+2}}
\label{sec:beyond}

Table~\ref{tab:ranges} shows that structural limitations
are visible also beyond dimension $k+2$, and most of them are explained by
Corollary~\ref{cor:lacunary}: for instance, the bound gives $f_5\leq0.0123$ for
\mwc[192], whereas the observed maximum is $0.0118$. The remaining limitations
are explained by the following generalization of Propositions~\ref{prop:twodim}
and~\ref{prop:threedim} (see also~\cite[Proposition~2]{CoLLSCLCSRASG}):

\begin{proposition}
\label{prop:general}
Assume that $a$ is coprime with $a_0$ (and thus with $m$; this is always true
for MWC generators). Let $J$ be a set of at least two pairs
$(i,j)$ of nonnegative integers such that the positions $ik+j$ are distinct and
smaller than $t$, let $d=\max_{(i,j)\in J}i$, and let $\bm
v_J=(a^{d-i}a_0^ib^j)_{(i,j)\in J}$. Then,
\[
\nu_t\leq\gamma_{|J|-1}^{1/2}\bigl(\|\bm v_J\|/\gcd(\bm v_J)\bigr)^{1/(|J|-1)}.
\]
\end{proposition}
\begin{proof}
Since $ab^k\equiv a_0\pmod m$, we have $a^db^{ik+j}\equiv a^{d-i}a_0^ib^j\pmod m$.
Thus, if $\bm z\in\Z^J$ is orthogonal to $\bm v_J$, the vector of $\Z^t$ with
$z_{ij}$ in position $ik+j$, and zeroes elsewhere, belongs to $L^*_t(b)$, as
$a^d\sum_{(i,j)\in J}z_{ij}b^{ik+j}\equiv\bm z\cdot\bm v_J=0\pmod m$ and $a$ is
invertible modulo $m$. The integer vectors orthogonal to $\bm v_J$ form a lattice
of rank $|J|-1$ and determinant $\|\bm v_J\|/\gcd(\bm v_J)$, and we conclude as
in the proof of Proposition~\ref{prop:twodim}.
\end{proof}

For example, for $k\geq3$ the set $J=\{\,(i,j)\mid i\leq1,j\leq2\,\}$, with
positions $\{0,1,2,k,k+1,k+2\}$, gives $f_6\leq0.0496$ for \gmwc[256], against
an observed maximum of $0.046$, and for $k\geq2$ the set
$J=\{\,(i,j)\mid 0\leq j\leq i\leq 2\,\}$, with positions
$\{0,k,k+1,2k,2k+1,2k+2\}$, gives $f_7\leq0.287$ for \mwc[192], against an
observed maximum of $0.267$. For $k=1$ the positions of these families collide,
which is consistent with the absence of a structural limitation on $f_3$ for
GMWC generators with lag one (Table~\ref{tab:ranges}).

Table~\ref{tab:spacings} reports, for the multipliers suggested in
Section~\ref{sec:results}, the distance between hyperplanes in dimension $k+2$;
the worst three-dimensional projection with indices up to $24$, that is, the set
$I=\{0,i_2,i_3\}$ with $i_3\leq 24$ minimizing $\nu_I$, and the worst
four-dimensional projection with indices up to $12$, both computed exactly with
\texttt{fplll}; and the minimum of $\nu_t/m^{1/t}$ for $9\leq t\leq 16$,
which measures the quality of the lattice structure in higher dimensions
without using Hermite constants (for the densest lattices,
$\nu_t/m^{1/t}=\gamma_t^{1/2}>1$). In the table, the worst three-dimensional
projection of every MWC generator is exactly the one with indices $\{0,1,k+1\}$
of Corollary~\ref{cor:lacunary}; for the GMWC generators with $k\geq2$, the worst
four-dimensional projection is the one with indices $\{0,1,k,k+1\}$ of
Proposition~\ref{prop:threedim}, with a distance between hyperplanes of about
$b^{-2/3}$. The low values of
$\nu_t/m^{1/t}$ for \mwc[256] and \mwc[128/32] ($k=3$) are attained in dimension
$9=2k+3$, the smallest dimension containing the positions
$\{0,k,k+1,2k,2k+1,2k+2\}$ of the second family above.

\begin{table}
\caption{\label{tab:spacings}Distances between hyperplanes for the multipliers
suggested in Section~\ref{sec:results}: $\log_2 m$; $\log_2\nu_{k+2}$; the
indices $I$ and $\log_2\nu_I$ of the worst three-dimensional projection with
indices up to $24$ and of the worst four-dimensional projection with indices up
to $12$; and $\min_{9\leq t\leq 16}\nu_t/m^{1/t}$. Indices refer to outputs.}
\renewcommand{\arraystretch}{1.1}
\small
\begin{tabular}{lrrlrlrr}
& & & \multicolumn{2}{c}{Worst 3-D} & \multicolumn{2}{c}{Worst 4-D} & \\
\cmidrule(rl){4-5}\cmidrule(rl){6-7}
Generator & \multicolumn{1}{c}{$\log_2 m$} & \multicolumn{1}{c}{$\log_2\nu_{k+2}$} & \multicolumn{1}{c}{$I$} & \multicolumn{1}{c}{$\log_2\nu_I$} & \multicolumn{1}{c}{$I$} & \multicolumn{1}{c}{$\log_2\nu_I$} & \multicolumn{1}{c}{$\min\nu_t/m^{1/t}$}\\
\hline
\mwc[128] & $128$ & $32.28$ & $\{0,1,2\}$ & $32.28$ & $\{0,3,9,12\}$ & $29.60$ & $0.76$\\
\mwc[192] & $192$ & $32.31$ & $\{0,1,3\}$ & $32.31$ & $\{0,1,2,3\}$ & $32.31$ & $0.74$\\
\mwc[256] & $256$ & $32.29$ & $\{0,1,4\}$ & $32.29$ & $\{0,1,2,4\}$ & $32.29$ & $0.10$\\
\gmwc[128] & $128$ & $42.73$ & $\{0,12,17\}$ & $38.99$ & $\{0,4,7,12\}$ & $28.67$ & $0.83$\\
\gmwc[256] & $256$ & $42.70$ & $\{0,1,4\}$ & $55.20$ & $\{0,1,3,4\}$ & $42.70$ & $0.81$\\
\mwc[64/32] & $64$ & $16.32$ & $\{0,1,2\}$ & $16.32$ & $\{0,4,9,10\}$ & $13.86$ & $0.77$\\
\mwc[128/32] & $128$ & $16.29$ & $\{0,1,4\}$ & $16.29$ & $\{0,1,2,4\}$ & $16.29$ & $0.19$\\
\gmwc[64/32] & $64$ & $21.43$ & $\{0,10,11\}$ & $19.08$ & $\{0,2,5,12\}$ & $14.09$ & $0.88$\\
\gmwc[128/32] & $128$ & $21.40$ & $\{0,1,4\}$ & $23.68$ & $\{0,1,3,4\}$ & $21.40$ & $0.97$\\
\end{tabular}
\end{table}

\subsection{Detecting lattice structure}
\label{sec:detect}

For $b=2^{32}$ the defect of Corollary~\ref{cor:lacunary} is easily detectable: Thomas already observed that
the multipliers of MWC generators with $b=2^{32}$ and $k=1$ give a poor
distribution of successive triples, and that such generators fail tests on
triples.\footnote{\label{fn:thomas}\url{http://cas.ee.ic.ac.uk/people/dt10/research/rngs-gpu-mwc64x.html} (2011).} A textbook three-dimensional
birthday-spacings test~\cite{KnuACPII,LESPBSTCFRNG} using the upper $20$ bits
of three successive outputs, with $10$ repetitions of $2^{21}$ points (a few
seconds of computation), finds $1\,029$ collisions for
\mwc[64/32] (with the multiplier suggested in Section~\ref{sec:results}),
whereas $20$ are expected.
On \mwc[128/32] the same test finds the expected number of collisions on
successive triples, but $1\,088$ on the lacunary triples
$(x_n,x_{n+1},x_{n+4})$, whereas \gmwc[64/32] and \gmwc[128/32] pass the test on
both kinds of triples.

With a $64$-bit modulus, however, lattice structure is detectable quickly also
in generators without structural defects: on the triples
$(x_n,x_{n+10},x_{n+11})$, the worst three-dimensional projection of
\gmwc[64/32] in Table~\ref{tab:spacings}, the same test using the upper $22$
bits of each output and $2^{23}$ points per repetition (less than half a minute
of computation) finds $91$ collisions, whereas $20$ are expected.

For $b=2^{64}$ the defect of \mwc[128] is much harder to detect. For the
\mwc[128] multiplier of Table~\ref{tab:best-mwc}, we have
$\nu_3\approx 2^{32.28}$, so successive planes are about $13$ times farther apart than the side of a cell of a
three-dimensional test using the upper $36$ bits of each output (recall that
outputs differ from the points of the associated MCG by less than $2^{-64}$).
Indeed, a three-dimensional birthday-spacings test~\cite{VigCB} on the upper $36$ bits of
$1.28\times10^{11}$ triples of successive outputs (about $3$\,TB) finds $81$
collisions, whereas $1.62$ are expected. A test of this size must
process terabytes of output independently of the available memory, which the
implementation of~\cite{VigCB} does by trading memory for time, and must use
cell numbers of more than $64$ bits: for example, the birthday-spacings test of
TestU01~\cite{LESTU01} accepts at most $2^{64}$ cells, that is, about $21$ bits
per coordinate in three dimensions, and thus cannot see planes at distance
about $2^{-32}$.

By Corollary~\ref{cor:lacunary}, the lacunary triples
$(x_n,x_{n+1},x_{n+k+1})$ of \mwc[192] and \mwc[256] lie in planes at
essentially the same distance (Table~\ref{tab:spacings}), and indeed the
same test finds $81$ collisions on the triples $(x_n,x_{n+1},x_{n+3})$
of \mwc[192] (about $4$\,TB of output), and $78$ on the triples
$(x_n,x_{n+1},x_{n+4})$ of \mwc[256] (about $5$\,TB). The $(k+2)$-tuples of
successive outputs of \mwc[192] and \mwc[256] lie in families of hyperplanes at
a similar distance, too, but a birthday-spacings test with the same resolution
would need about $2^{12}$ and $2^{24}$ times more points,
respectively.\footnote{\label{fn:scaling}The expected number of collisions is proportional to
$n^3$ divided by the number of cells, where $n$ is the number of points; using
$36$ bits per coordinate in dimension $k+2$ multiplies the number of cells by
$2^{36(k-1)}$, and thus $n$ by $2^{12(k-1)}$.}

Distances in different dimensions can be compared through the cost of such a
test: since the number of points grows as the cube root of the number of cells
(footnote~\ref{fn:scaling}), at the same relative resolution a test on planes at
distance $d$ in dimension $t$ needs a number of points proportional to
$d^{-t/3}$. For example, the hyperplanes of the worst four-dimensional
projection of \gmwc[128] in Table~\ref{tab:spacings} are farther apart than the
planes of the successive triples of \mwc[128], but a test on either of the worst
projections of \gmwc[128] needs at least $60$ times more points than a test on
the successive triples of \mwc[128].

Note that the congruence defining the planes holds modulo $b$, up to a small
error, so the defect is visible also in the lower bits: on the lowest bits,
the tests above find essentially the same numbers of collisions as on the upper
bits.

\section{Scrambling}
\label{sec:x}

The lattice structure studied in Section~\ref{sec:spectral} is a property of the digits
$x_n$, which are, up to an error smaller than $1/b$, the points of the
associated MCG. Thomas's MWC64X generator (footnote~\ref{fn:thomas}), an MWC
generator with $b=2^{32}$ and $k=1$ designed for GPUs, returns instead
$x_n\oplus c_n$, the bitwise exclusive or of the digit and of the carry computed
together with it by~(\ref{eq:rec}), that is, of the lower and upper halves of
$\tau_n=ax_{n-k}+c_{n-1}$; note that $c_n$ and $x_n$ are the two most
significant base-$b$ digits of the state integer of the state
$(x_{n-k+1},\dots,x_n;c_n)$, and for $k=1$ they are the whole state.

For GMWC generators, $x_n\oplus c_n$ would be biased: by~(\ref{eq:rec}),
$c_n=\lfloor\tau_n/b\rfloor+\lceil ex_n/b\rceil$, so the range of the
carry is shifted by about $ex_n/b$, the upper bits of $x_n$ and $c_n$ are
positively correlated, and the most significant bit of $x_n\oplus c_n$ is one
with probability about $1/2-e/(4b)$. The bias is negligible for \gmwc[128]
(for our multiplier, $e<2^{31}$), but it is about $2\times10^{-4}$ for \gmwc[256], and for
\gmwc[64/32] PractRand detects it after $4$\,GB of output. We thus define a
scrambler that generalizes Thomas's, returning
$x_n\oplus\lfloor\tau_n/b\rfloor$: since $\lfloor\tau_n/b\rfloor=c_n$ for
MWC generators, it degenerates to Thomas's scrambler, so we simply call it
\emph{X}, and we denote scrambled generators by the suffix X (e.g.,
\mwc[192X], \gmwc[128X], and \mwc[64X/32]). For GMWC generators,
$\tau_n=c_nb-ex_n$ is the state integer when $k=1$, and differs by at most $e$
from the integer formed by its two most significant base-$b$ digits otherwise;
since $x_n=A\tau_n\bmod b$ depends only on $\tau_n\bmod b$, X is a function of
$\tau_n$ alone, as for MWC generators. Since
$\lfloor\tau_n/b\rfloor\leq c_n$ and the carry fits in a $w$-bit word
(Section~\ref{sec:seeding}), X costs a single instruction, which is not on the
loop-carried dependency chain (Section~\ref{sec:speed}; see
Figure~\ref{fig:code}).

Thomas motivates his scrambler by noting that the carry changes completely at each step,
and that the exclusive or is an operation of a different kind than those of the
recurrence. More specifically, $0\leq x_{n-k}-\lfloor\tau_n/b\rfloor\leq
b-a$, so when $a$ is close to $b$ the upper bits of $\lfloor\tau_n/b\rfloor$
(for an MWC generator, $c_n$) are, up to a borrow, those of $x_{n-k}$ (for our
multipliers, $a\geq b-b/2^8$, so at least the upper eight bits): X combines
$x_{n-k}$ with $x_n$ by an operation that does not preserve the linear relations with
small integer coefficients that define the families of hyperplanes of the
spectral test. For the same reason, X is useful only if the carry has about as
many bits as the digit: the carry of \texttt{mwc59} (Section~\ref{sec:small})
has $27$ bits, so X would leave the upper five bits of its $32$-bit digit
unchanged.

Empirically, X
removes from all bits every defect we could detect. For $b=2^{32}$, with X the test of Section~\ref{sec:detect} finds
the expected number of collisions, on the upper and on the lower bits, on the triples on
which \mwc[64/32] and \mwc[128/32] fail, and X hides also the lattice structure
due to the $64$-bit modulus of \gmwc[64/32]. More generally, on the upper and on the lower bits of the
outputs with indices $\{0,i,j\}$, $0<i<j\leq8$, of these three generators and
of \gmwc[128/32] with X, all $224$ counts are within three standard deviations
of the expected value.
For $b=2^{64}$, on the upper and on the lowest $36$ bits of the triples on which
\mwc[128], \mwc[192], and \mwc[256] fail, the test of Section~\ref{sec:detect}
finds with X at most $4$ collisions, whereas $1.62$ are expected.

As for every scrambler, it is a good idea to choose the best possible
underlying generator, and we still select multipliers using the spectral test,
in the same way in which the scrambled linear generators of Blackman and
Vigna~\cite{BlVSLPNG} use underlying generators with good $\F_2$-linear
properties.

\section{Selecting multipliers}
\label{sec:criterion}

\subsection{Classical criteria}

When choosing a multiplier for an MCG, the figures of merit $f_2$, $f_3$,
\dots, $f_T$ are usually summarized by a single criterion. The most common
choice is the worst-case figure of merit $M_T=\min_{2\leq t\leq T}f_t$, which
was used by Fishman and Moore~\cite{FiMEAMCRNG} and later by L'Ecuyer and
coauthors~\cite{LBCSGMRRNG,LEcTLCGDSGLS}, who also considered worst cases over
projections with lacunary indices~\cite{LeLVRLR}.

Steele and Vigna~\cite[Section~9]{StVCESGMCPNGjr} observed that, when selecting
among many candidates, the minimum is a questionable choice even for ordinary
congruential generators, and gave the following example of two multipliers for
a $64$-bit LCG, with their figures of merit $f_2$, $f_3$, \dots, $f_8$:
\begin{center}
\small
\begin{tabular}{lrrrrrrr}
\texttt{0xe2e19bb27190da6d} & $0.791216$ & $0.771300$ & $0.791569$ & $0.777944$ & $0.773526$ & $0.777463$ & $0.766073$\\
\texttt{0xe73d20db8e96d2cd} & $0.941271$ & $0.883251$ & $0.854317$ & $0.825078$ & $0.803654$ & $0.781546$ & $0.766043$
\end{tabular}
\end{center}
The minimum chooses the first multiplier, whose $f_8$ is larger in the fifth
decimal place, although the second multiplier is much better in all other
dimensions. Steele and Vigna thus consider as equivalent all multipliers whose
minimum figure of merit is in the top $0.1\%$ of the candidates, and among them
choose the one with the best \emph{harmonic score}
\[
H_T=\frac1{\sum_{t=2}^T 1/(t-1)}\sum_{t=2}^T\frac{f_t}{t-1},
\]
which gives more weight to low dimensions.

For (G)MWC generators both criteria break down. Figures of merit in dimensions
up to $k+1$ are essentially fixed (Proposition~\ref{prop:low}); for MWC
generators, and for GMWC generators with $k\geq 2$, the figure of merit in
dimension $k+2$ (and in a few more dimensions) is structurally bounded by a very
small quantity (Propositions~\ref{prop:twodim} and~\ref{prop:threedim}). Thus,
even ignoring the fixed dimensions, maximizing the worst-case figure of merit
means optimizing a single, structurally bad dimension, with no control on the
others. The harmonic score has the opposite problem: the figures of merit in the
bad dimensions are tiny, so they contribute nothing to the score.

The problem is that the absolute values of the figures of merit in different
dimensions are not comparable: for \mwc[128], a value of $f_3$ equal to
$6\times 10^{-4}$ is excellent, whereas a value of $f_6$ equal to $0.6$ is
mediocre.

\subsection{The weighted geometric mean}

Ratios of figures of merit, however, are comparable: a candidate whose $f_3$ is $10\%$ larger than
that of another candidate is better in dimension three to the same extent as a
candidate whose $f_6$ is $10\%$ larger is better in dimension six. We thus look
for a way of ranking candidates that does not depend on the normalization of
each figure of merit (by Hermite constants, by the structural bounds of
Section~\ref{sec:spectral}, or by the values observed on a set of candidates),
and in which the comparison between two candidates depends only on their
figures of merit. These requirements essentially determine the ranking, by a
classical characterization from social choice theory, where the same invariance
is required for utilities measured on ratio scales (Fleming and
Wallace~\cite{FlWHNLSCWSBR} use the same argument to summarize normalized
benchmark results by their geometric mean):

\begin{proposition}[{\cite[Theorems~4.2 and~4.17]{dAGSWFIC}}]
\label{prop:geomean}
Let $\succeq$ be a total preorder on $(0\..\infty)^n$ such that:
\begin{enumerate}
\item for every $\bm y$, the sets $\{\,\bm x\mid\bm x\succeq\bm y\,\}$ and
$\{\,\bm x\mid\bm y\succeq\bm x\,\}$ are closed;
\item if $x_t>y_t$ for every $t$, then $\bm x\succ\bm y$;
\item if $\bm x\succeq\bm y$, then $(c_1x_1,\dots,c_nx_n)\succeq(c_1y_1,\dots,c_ny_n)$
for every $\bm c\in(0\..\infty)^n$.
\end{enumerate}
Then, there are $w_1$, \dots, $w_n\geq0$, not all zero, such that $\bm
x\succeq\bm y$ if and only if $\prod_tx_t^{w_t}\geq\prod_ty_t^{w_t}$.
\end{proposition}
\begin{proof}
Taking logarithms, rescalings become translations, and the statement
becomes the characterization of weighted utilitarianism of d'Aspremont and
Gevers~\cite[Theorem~4.2]{dAGSWFIC}.
\end{proof}

Proposition~\ref{prop:geomean} leaves us free to choose the weights. We use the
weights $1/(t-1)$ of the harmonic score $H_T$, which give more importance to low
dimensions, and rank candidates by the \emph{weighted geometric mean}
\[
G_T=\Bigl(\prod_{t=2}^Tf_t^{1/(t-1)}\Bigr)^{1/\sum_{t=2}^T 1/(t-1)},
\]
that is, by the harmonic score of the logarithms of the figures of merit. We use
$T=8$, as is common in the literature~\cite{LEcTLCGDSGLS,StVCESGMCPNGjr}. Since the ranking does not depend on the normalization,
Hermite constants play no role, and in dimensions in which they are not known
one can use $\nu_t/m^{1/t}$ in place of $f_t$.

The harmonic score itself, instead, does not satisfy the third condition of
Proposition~\ref{prop:geomean}: its choices depend on the normalization of the
figures of merit. For example, dividing each figure of merit by its maximum over
the candidates, rather than by $\nu^*_t(m)$, changes the choice of $H_8$ in six
of our nine searches.

The weighted geometric mean has a natural interpretation:
$-\log_2f_t=\log_2\nu^*_t(m)-\log_2\nu_t$ is the number of bits of resolution
lost in dimension~$t$ with respect to an optimal lattice with the same density,
so maximizing $G_T$ minimizes a weighted total of the number of bits lost. In
this view, $M_T$ minimizes the largest loss, which for MWC generators is
determined by their structure (\mwc[128] loses more than $10$ bits in dimension
three, whatever the multiplier), whereas $H_T$ ignores the losses in dimensions
in which $f_t$ is very small.

The candidate maximizing $G_T$ is Pareto optimal, as $G_T$ is strictly
increasing in each figure of merit. Figures of merit essentially fixed by the
structure of the generator multiply $G_T$ by essentially the same factor for all
candidates, so they need not be excluded (on our data, excluding them does not
change the choice). Unlike $M_T$, $G_T$ compensates a loss in a dimension with
gains in others, but only in relative terms: on our data, every figure of merit
of the chosen candidates is at least $72\%$ of the best value found in its
dimension ($82\%$ except for \gmwc[128], see Section~\ref{sec:comparison}). The
choice depends little on the weights: equal weights would change it only for
\gmwc[128], and weights $t$ (i.e., equal weights on the logarithms of Knuth's
figures of merit $\mu_t\propto f_t^t$~\cite{KnuACPII}) only for \gmwc[128] and
\gmwc[64/32].

The weighted geometric means of the five best candidates of each search differ
by less than $1.3\%$, so they should be considered essentially equivalent (in
Section~\ref{sec:results} we report all of them).

\subsection{Comparison}
\label{sec:comparison}

On our searches for $b=2^{64}$, we compared the choices of $G_8$ with those of
$M_8$ (computed on dimensions $k+2$ to $8$), of $H_8$ (computed, as
in~\cite{StVCESGMCPNGjr}, on dimensions $2$ to $8$), and of the approach of
Steele and Vigna (SV: the best $H_8$ among the top $0.1\%$ candidates by $M_8$).
When some dimensions are structurally bounded, $M_8$ optimizes just the bounded
dimension: for \mwc[128], its choice has the largest $f_3$, but $f_8=0.407$,
against $0.728$ for the choice of $G_8$. The multiplier of Thomas's MWC64X
(footnote~\ref{fn:thomas}) is an instance with $b=2^{32}$: among our \mwc[64/32]
candidates, it is exactly the choice of $M_8$, as it has the largest $f_3$, but
$f_5=0.101$, and
it ranks $15\,741$st among the $22\,982$ candidates by $G_8$. SV avoids the
worst choices, but for MWC generators the top $0.1\%$ by $M_8$ contains just the
candidates with the best $f_{k+2}$, and the chosen candidate may be mediocre in
some other dimension (for \mwc[192], $f_7=0.204$, against $0.256$ for the choice
of $G_8$). $H_8$ often chooses the same
candidate as $G_8$, or an essentially equivalent one, but it can choose
candidates that are poor in the structurally bounded dimensions: for
\gmwc[256], its choice has $f_5=1.55\times10^{-3}$, close to the median value
$1.52\times10^{-3}$ over all candidates, against $2.24\times10^{-3}$ for $G_8$.

When no dimension is structurally bounded, as for \gmwc[128], $M_8$ is a sensible
criterion: all figures of merit of its choice in dimensions $3$ to $8$ are
between $0.77$ and $0.80$,
whereas the choice of $G_8$ (Table~\ref{tab:best-gmwc}) has much better $f_3$
and $f_4$, but $f_8=0.613$. The weights of $G_8$ favor
the latter, but the choice between the two candidates is a matter of taste.

\section{The search}
\label{sec:search}

We want the modulus $m$ to be a safe prime, so that the period is $(m-1)/2$,
and $a$ to be as close as possible to $b$, so that the period and the figure of
merit in dimension $k+1$, which is approximately
$(a/b)^{k/(k+1)}/\gamma_{k+1}^{1/2}$, are as large as possible. We thus consider multipliers $a$ whose upper $8$ bits are all ones,
that is, $a\in[b-b/2^8\..b)$. For GMWC generators we have the additional
parameter $e$, which must be odd and satisfy $a+e<b$, so that the carry fits a
single word (see Section~\ref{sec:seeding}); thus, $e<2^{w-8}$. Optionally, $e$ can be restricted to fewer bits, but it should be larger than
$b^{1/3}$, to avoid the structural bound of Proposition~\ref{prop:twodim} on
$f_{k+2}$ (for $k\geq 2$, the weaker bound of Proposition~\ref{prop:threedim},
essentially independent of $e$, remains). We
restricted $e$ to $32$ bits for \gmwc[128]; in the other GMWC searches, $e$ has
the maximum allowed size (Table~\ref{tab:search}).

For \mwc[64/32] and \mwc[128/32] we enumerate exhaustively all multipliers in
$[2^{32}-2^{24}\..2^{32})$. In all other cases, we generate random candidate
multipliers using a \xoshiropp[256] generator~\cite{BlVSLPNG}, setting the upper
$8$ bits; in the GMWC case, we generate $e$ in the same way, masking it to the
desired number of bits and setting the lowest bit; pairs with $a+e\geq b$ are
discarded. We test whether $m$ and $(m-1)/2$ are prime using the Miller--Rabin
test of the NTL library~\cite{ShoNTL}. Candidates passing the test are then
subjected to the spectral test in dimensions $2$ to $8$ using
LatticeTester~\cite{LEcLTG}, which computes the shortest vectors of the dual lattice exactly (in the
Euclidean norm) by branch-and-bound after lattice
reduction~\cite{LECILSTMRLRNG}. As a sanity check, we recomputed the figures of
merit of all the candidates appearing in our tables using \texttt{fplll},
obtaining identical results to the six significant digits printed by
LatticeTester; our tables use the exact values, to avoid rounding twice. We
also proved the primality of $m$ and $(m-1)/2$ for the five best candidates of
each search (using the APR-CL test of PARI/GP through Sage), so the period of
the generators we suggest is exactly $(m-1)/2$.

\begin{table}
\caption{\label{tab:search}Summary of our searches: base, lag, number of
candidates (safe-prime moduli on which the spectral test has been run), maximum
number of bits of $e$ (GMWC only).}
\renewcommand{\arraystretch}{1.1}
\begin{tabular}{lcrrr}
Generator & $b$ & \multicolumn{1}{c}{$k$} & \multicolumn{1}{c}{Candidates} & \multicolumn{1}{c}{Bits of $e$}\\
\hline
\mwc[128] & $2^{64}$ & $1$ & $11\,074\,098$ & \multicolumn{1}{c}{---}\\
\mwc[192] & $2^{64}$ & $2$ & $7\,009\,216$ & \multicolumn{1}{c}{---}\\
\mwc[256] & $2^{64}$ & $3$ & $1\,172\,795$ & \multicolumn{1}{c}{---}\\
\gmwc[128] & $2^{64}$ & $1$ & $10\,941\,185$ & $32$\\
\gmwc[256] & $2^{64}$ & $3$ & $1\,317\,821$ & $56$\\
\mwc[64/32] & $2^{32}$ & $1$ & $22\,982$ & \multicolumn{1}{c}{---}\\
\mwc[128/32] & $2^{32}$ & $3$ & $5\,601$ & \multicolumn{1}{c}{---}\\
\gmwc[64/32] & $2^{32}$ & $1$ & $10\,659\,050$ & $24$\\
\gmwc[128/32] & $2^{32}$ & $3$ & $10\,340\,733$ & $24$\\
\end{tabular}
\end{table}

\section{Suggested multipliers}
\label{sec:results}

Tables~\ref{tab:best-mwc}, \ref{tab:best-gmwc}, and~\ref{tab:best-32} report
the five best candidates by $G_8$ for each of our searches, together
with their figures of merit in dimensions $k+2$ to $8$ (in lower dimensions
figures of merit are essentially fixed, see Table~\ref{tab:ranges}); the top
candidates should be considered essentially equivalent (see
Section~\ref{sec:criterion}). For GMWC generators, the constant
$A=(-e)^{-1}\bmod b$ needed by the implementation can be computed, for instance,
by five steps of the Newton iteration $y\leftarrow y(2+ey)$ starting from
$y=-e$.

\begin{table}
\caption{\label{tab:best-mwc}The five best multipliers by $G_8$ for \mwc[128],
\mwc[192], and \mwc[256], with their figures of merit in dimensions $k+2$ to
$8$.}
\renewcommand{\arraystretch}{1.1}
\small
\begin{tabular}{lrrrrrr}
\multicolumn{1}{c}{$a$} & \multicolumn{1}{c}{$f_3$} & \multicolumn{1}{c}{$f_4$} & \multicolumn{1}{c}{$f_5$} & \multicolumn{1}{c}{$f_6$} & \multicolumn{1}{c}{$f_7$} & \multicolumn{1}{c}{$f_8$}\\
\hline
\multicolumn{7}{c}{\mwc[128]}\\
\hline
\texttt{0xff6da236274e252e} & $6.64\times 10^{-4}$ & $0.9158$ & $0.8345$ & $0.8359$ & $0.7713$ & $0.7284$\\
\texttt{0xff9d709b20cf6d2f} & $6.97\times 10^{-4}$ & $0.9385$ & $0.8714$ & $0.7453$ & $0.6993$ & $0.6993$\\
\texttt{0xff733b5223b28db3} & $6.94\times 10^{-4}$ & $0.9274$ & $0.8320$ & $0.7774$ & $0.7232$ & $0.7095$\\
\texttt{0xff94d7cb4e4860fd} & $6.92\times 10^{-4}$ & $0.9167$ & $0.8862$ & $0.7791$ & $0.7314$ & $0.6454$\\
\texttt{0xff653dd2f96f56c4} & $6.98\times 10^{-4}$ & $0.9382$ & $0.7705$ & $0.7670$ & $0.7661$ & $0.7246$\\
\hline
\end{tabular}

\medskip

\begin{tabular}{lrrrrr}
\multicolumn{1}{c}{$a$} & \multicolumn{1}{c}{$f_4$} & \multicolumn{1}{c}{$f_5$} & \multicolumn{1}{c}{$f_6$} & \multicolumn{1}{c}{$f_7$} & \multicolumn{1}{c}{$f_8$}\\
\hline
\multicolumn{6}{c}{\mwc[192]}\\
\hline
\texttt{0xff4e24dd418f5433} & $1.60\times 10^{-5}$ & $0.0112$ & $0.9007$ & $0.2563$ & $0.7936$\\
\texttt{0xff34e6efba450121} & $1.56\times 10^{-5}$ & $0.0115$ & $0.9278$ & $0.2476$ & $0.7437$\\
\texttt{0xff149c8cb7b203b8} & $1.57\times 10^{-5}$ & $0.0115$ & $0.9219$ & $0.2431$ & $0.7542$\\
\texttt{0xffc5fcc4f63c15e7} & $1.57\times 10^{-5}$ & $0.0117$ & $0.9395$ & $0.2333$ & $0.7424$\\
\texttt{0xff52883eb6de17c1} & $1.55\times 10^{-5}$ & $0.0116$ & $0.9310$ & $0.2455$ & $0.7409$\\
\hline
\end{tabular}

\medskip

\begin{tabular}{lrrrr}
\multicolumn{1}{c}{$a$} & \multicolumn{1}{c}{$f_5$} & \multicolumn{1}{c}{$f_6$} & \multicolumn{1}{c}{$f_7$} & \multicolumn{1}{c}{$f_8$}\\
\hline
\multicolumn{5}{c}{\mwc[256]}\\
\hline
\texttt{0xff32e7e2bdcb0a33} & $1.65\times 10^{-6}$ & $5.83\times 10^{-4}$ & $0.0382$ & $0.8641$\\
\texttt{0xff0e07044e7aeb95} & $1.65\times 10^{-6}$ & $5.81\times 10^{-4}$ & $0.0381$ & $0.8602$\\
\texttt{0xffa68b2009f3b90f} & $1.65\times 10^{-6}$ & $5.82\times 10^{-4}$ & $0.0381$ & $0.8631$\\
\texttt{0xff17f83a1f4af883} & $1.64\times 10^{-6}$ & $5.81\times 10^{-4}$ & $0.0381$ & $0.8615$\\
\texttt{0xff3f195afd6226e4} & $1.64\times 10^{-6}$ & $5.81\times 10^{-4}$ & $0.0381$ & $0.8615$\\
\hline
\end{tabular}
\end{table}

\begin{table}
\caption{\label{tab:best-gmwc}The five best pairs $(a,e)$ by $G_8$ for
\gmwc[128] and \gmwc[256], with their figures of merit in dimensions $k+2$ to
$8$.}
\renewcommand{\arraystretch}{1.1}
\small
\begin{tabular}{llrrrrrr}
\multicolumn{1}{c}{$a$} & \multicolumn{1}{c}{$e$} & \multicolumn{1}{c}{$f_3$} & \multicolumn{1}{c}{$f_4$} & \multicolumn{1}{c}{$f_5$} & \multicolumn{1}{c}{$f_6$} & \multicolumn{1}{c}{$f_7$} & \multicolumn{1}{c}{$f_8$}\\
\hline
\multicolumn{8}{c}{\gmwc[128]}\\
\hline
\texttt{0xfff0cfb4ed6366fa} & \texttt{0x68935f57} & $0.9330$ & $0.9013$ & $0.7864$ & $0.7746$ & $0.7705$ & $0.6132$\\
\texttt{0xffd80a0afb781ae9} & \texttt{0xb48a3fab} & $0.9149$ & $0.8483$ & $0.7818$ & $0.8037$ & $0.7875$ & $0.6899$\\
\texttt{0xffd492eb05addaa4} & \texttt{0x2d1020d3} & $0.9168$ & $0.8344$ & $0.7746$ & $0.7593$ & $0.8048$ & $0.6959$\\
\texttt{0xffaec0a9a043fbfa} & \texttt{0xd01f08a3} & $0.9561$ & $0.8250$ & $0.7035$ & $0.8279$ & $0.8020$ & $0.6484$\\
\texttt{0xff21354728271635} & \texttt{0x907a82ab} & $0.9294$ & $0.8978$ & $0.8130$ & $0.7648$ & $0.6336$ & $0.6655$\\
\hline
\end{tabular}

\medskip

\begin{tabular}{llrrrr}
\multicolumn{1}{c}{$a$} & \multicolumn{1}{c}{$e$} & \multicolumn{1}{c}{$f_5$} & \multicolumn{1}{c}{$f_6$} & \multicolumn{1}{c}{$f_7$} & \multicolumn{1}{c}{$f_8$}\\
\hline
\multicolumn{6}{c}{\gmwc[256]}\\
\hline
\texttt{0xff76e600f25ae046} & \texttt{0x34dabaf938e45b} & $2.24\times 10^{-3}$ & $0.0414$ & $0.7912$ & $0.7590$\\
\texttt{0xff2eba2993c1d518} & \texttt{0xeb1eb4ffc63e7} & $2.27\times 10^{-3}$ & $0.0420$ & $0.7450$ & $0.7741$\\
\texttt{0xff43a827aa5cad3e} & \texttt{0x8946f0524490a7} & $2.39\times 10^{-3}$ & $0.0396$ & $0.8255$ & $0.6704$\\
\texttt{0xff5feedbf607f9c5} & \texttt{0x821b4720047a0f} & $2.37\times 10^{-3}$ & $0.0381$ & $0.7920$ & $0.7381$\\
\texttt{0xff01354317d8c2f4} & \texttt{0xd5bd05c69bb347} & $2.25\times 10^{-3}$ & $0.0447$ & $0.7127$ & $0.7206$\\
\hline
\end{tabular}
\end{table}

\begin{table}
\caption{\label{tab:best-32}The five best multipliers (or pairs) by $G_8$ for
$b=2^{32}$, with their figures of merit in dimensions $k+2$ to $8$.}
\renewcommand{\arraystretch}{1.1}
\small
\begin{tabular}{lrrrrrr}
\multicolumn{1}{c}{$a$} & \multicolumn{1}{c}{$f_3$} & \multicolumn{1}{c}{$f_4$} & \multicolumn{1}{c}{$f_5$} & \multicolumn{1}{c}{$f_6$} & \multicolumn{1}{c}{$f_7$} & \multicolumn{1}{c}{$f_8$}\\
\hline
\multicolumn{7}{c}{\mwc[64/32]}\\
\hline
\texttt{0xfff145f8} & $0.0275$ & $0.8957$ & $0.7406$ & $0.8239$ & $0.7191$ & $0.6889$\\
\texttt{0xff111cb4} & $0.0277$ & $0.9060$ & $0.8239$ & $0.7542$ & $0.7348$ & $0.5805$\\
\texttt{0xffdb8e77} & $0.0279$ & $0.9325$ & $0.7737$ & $0.7245$ & $0.6627$ & $0.6848$\\
\texttt{0xff874455} & $0.0269$ & $0.9075$ & $0.8163$ & $0.7375$ & $0.6871$ & $0.7019$\\
\texttt{0xffe09073} & $0.0263$ & $0.9098$ & $0.8261$ & $0.7969$ & $0.7158$ & $0.5697$\\
\hline
\end{tabular}

\medskip

\begin{tabular}{lrrrr}
\multicolumn{1}{c}{$a$} & \multicolumn{1}{c}{$f_5$} & \multicolumn{1}{c}{$f_6$} & \multicolumn{1}{c}{$f_7$} & \multicolumn{1}{c}{$f_8$}\\
\hline
\multicolumn{5}{c}{\mwc[128/32]}\\
\hline
\texttt{0xffcccc66} & $1.28\times 10^{-3}$ & $0.0233$ & $0.1847$ & $0.8566$\\
\texttt{0xff7df6e5} & $1.27\times 10^{-3}$ & $0.0233$ & $0.1845$ & $0.8563$\\
\texttt{0xff555f96} & $1.27\times 10^{-3}$ & $0.0233$ & $0.1844$ & $0.8558$\\
\texttt{0xff2512f8} & $1.26\times 10^{-3}$ & $0.0232$ & $0.1840$ & $0.8537$\\
\texttt{0xffafc94a} & $1.28\times 10^{-3}$ & $0.0231$ & $0.1835$ & $0.8508$\\
\hline
\end{tabular}

\medskip

\begin{tabular}{llrrrrrr}
\multicolumn{1}{c}{$a$} & \multicolumn{1}{c}{$e$} & \multicolumn{1}{c}{$f_3$} & \multicolumn{1}{c}{$f_4$} & \multicolumn{1}{c}{$f_5$} & \multicolumn{1}{c}{$f_6$} & \multicolumn{1}{c}{$f_7$} & \multicolumn{1}{c}{$f_8$}\\
\hline
\multicolumn{8}{c}{\gmwc[64/32]}\\
\hline
\texttt{0xff66966f} & \texttt{0x651e9f} & $0.9514$ & $0.8246$ & $0.8271$ & $0.7840$ & $0.7437$ & $0.7605$\\
\texttt{0xff4bce10} & \texttt{0x2ec607} & $0.9437$ & $0.8981$ & $0.7672$ & $0.7629$ & $0.7419$ & $0.7494$\\
\texttt{0xffb684f6} & \texttt{0x150f3f} & $0.9534$ & $0.8867$ & $0.8385$ & $0.6844$ & $0.7200$ & $0.6418$\\
\texttt{0xff3886d5} & \texttt{0x62e507} & $0.9569$ & $0.7888$ & $0.8551$ & $0.7871$ & $0.6590$ & $0.7160$\\
\texttt{0xff0a55d9} & \texttt{0x97b8cf} & $0.9069$ & $0.7951$ & $0.8156$ & $0.8033$ & $0.7973$ & $0.7122$\\
\hline
\end{tabular}

\medskip

\begin{tabular}{llrrrr}
\multicolumn{1}{c}{$a$} & \multicolumn{1}{c}{$e$} & \multicolumn{1}{c}{$f_5$} & \multicolumn{1}{c}{$f_6$} & \multicolumn{1}{c}{$f_7$} & \multicolumn{1}{c}{$f_8$}\\
\hline
\multicolumn{6}{c}{\gmwc[128/32]}\\
\hline
\texttt{0xff6b46e1} & \texttt{0x8e02fb} & $0.0442$ & $0.1894$ & $0.7826$ & $0.7635$\\
\texttt{0xff37e8ef} & \texttt{0x19aa97} & $0.0424$ & $0.1927$ & $0.7756$ & $0.7460$\\
\texttt{0xffbb32ec} & \texttt{0x19a4c3} & $0.0440$ & $0.1864$ & $0.7979$ & $0.7097$\\
\texttt{0xff4e0157} & \texttt{0x231523} & $0.0437$ & $0.1928$ & $0.7820$ & $0.6940$\\
\texttt{0xffb21299} & \texttt{0x40303} & $0.0425$ & $0.1856$ & $0.7740$ & $0.7749$\\
\hline
\end{tabular}
\end{table}

\subsection{A small-state generator}
\label{sec:small}

Dynamic languages often represent small integers in a machine word with a few
tag bits, so arithmetic is fast only if all intermediate results fit in, say,
$59$ or $60$ bits, and the multiplier must be much smaller than the base. The
\texttt{mwc59} generator of Erlang/OTP~\cite{OTPrand} uses $b=2^{32}$, $k=1$,
and $a=\texttt{0x7fa6502}$, so $m=a\cdot2^{32}-1$ is a $59$-bit safe prime; the
multiplier was chosen by figures of merit, and then by PractRand, among those in
$[2^{27}-2^{20}\..2^{27})$ for which $m$ is a safe prime. When $a$ is much
smaller than $b$, the shortest vector in dimension two is $\bm u$: for
\texttt{mwc59}, $f_2=0.164$, and $f_3=0.074$ is very close to the bound $0.075$
of Proposition~\ref{prop:twodim}, so successive triples lie in planes at
distance about $2^{-16}$, as for \mwc[64/32]. Indeed, the documentation of
\texttt{mwc59} states that its raw output has detectable problems in two- and
three-dimensional birthday-spacings and collision tests, and the functions
returning outputs apply a small xorshift-based scrambler~\cite{OTPrand}.

\section{Testing, speed, and recommendations}
\label{sec:evaluation}

\subsection{Statistical testing}
\label{sec:testing}

We tested the generators of Tables~\ref{tab:best-mwc}, \ref{tab:best-gmwc},
and~\ref{tab:best-32}, plain and with the scrambler X, using the first
multiplier (or pair) of each table, with
PractRand~\cite{DotPR} (version 0.96, expanded test set, extra folding), up to
$32$\,TB of output. With X, no failures were detected. Plain generators with
$b=2^{64}$ and plain \gmwc[128/32] pass the tests, too, whereas plain
\mwc[64/32] and \gmwc[64/32] fail after $8$\,GB of output, and plain
\mwc[128/32] after $4$\,TB.

PractRand does not detect the defects of plain \mwc[128], \mwc[192], and
\mwc[256], whereas the specialized birthday-spacings test of
Section~\ref{sec:detect} detects them using $3$ to $5$\,TB of output:
general-purpose test suites have little power against lattice structure at this
scale, which is precisely the reason why multipliers should be selected using
the spectral test. With $b=2^{32}$, instead, hyperplanes are much farther apart
(Table~\ref{tab:spacings}), and PractRand detects the lattice structure, too:
plain \mwc[128/32] fails, whereas plain \mwc[128], which has the same state
size, passes. Passing PractRand shows only that the generators have no gross
defect unrelated to their lattice structure.

\subsection{Speed}
\label{sec:speed}

Table~\ref{tab:speed} reports the time needed to generate a $64$-bit output
for our generators with $b=2^{64}$, plain and with X, and for a few other
popular generators. We include also \gmwc[192] ($k=2$), which was not part of our
search, using $a=\texttt{0xffaf1260e39d9b7f}$ and $e=\texttt{0xb4fc961f}$ (a
safe-prime modulus; the speed does not depend on the multipliers). As a
baseline, we include an MCG with prime modulus $2^{127}-1$, as
\mwc[128] simulates an MCG with a prime modulus of about the same size
(Theorem~\ref{teo:mcg}). Its multiplier has $64$ bits, and it is
larger than $\sqrt{2^{127}-1}$~\cite{StVCESGMCPNGjr}: each step needs two
$64$-bit multiplications and a reduction modulo $2^{127}-1$, which needs only
shifts and additions~\cite{PRBCLPNG}. AES-CTR encrypts a counter with AES-128 using the dedicated
AES-NI instructions of the processor: it encrypts four counter blocks at a time,
so that the rounds of different blocks are independent and can be executed in
parallel, and each group of blocks yields eight $64$-bit outputs. 

We used the harness of~\cite{BlVSLPNG}: each generator is
called $3\times10^9$ times in a loop that accumulates its outputs (by addition
for our generators, as X ends with a xor, and by xor for the others), with loop
unrolling disabled, and we report the minimum user time per output over five
runs.

\begin{table}
\caption{\label{tab:speed}Time in nanoseconds to generate a $64$-bit output,
plain and with the scrambler X (computed on the current state, except for X on
GMWC generators, which uses $x_n$ and $\tau_n$),
on an AMD EPYC~9R45 CPU (Zen~5) with frequency fixed at
$4.5$\,GHz, compiled by GCC~16.2.0 and Clang~23.1.3 with options \texttt{-O3
-march=native -fno-unroll-loops} (GMWC generators compiled by GCC also with
\texttt{-mno-avx512vl}).}
\renewcommand{\arraystretch}{1.1}
\small
\begin{tabular}{lrrrr}
 & \multicolumn{2}{c}{GCC} & \multicolumn{2}{c}{Clang}\\
\cmidrule(rl){2-3}\cmidrule(rl){4-5}
Generator & \multicolumn{1}{c}{Plain} & \multicolumn{1}{c}{X} & \multicolumn{1}{c}{Plain} & \multicolumn{1}{c}{X}\\
\hline
\mwc[128] & $0.93$ & $0.94$ & $0.90$ & $0.95$\\
\mwc[192] & $0.52$ & $0.53$ & $0.52$ & $0.51$\\
\mwc[256] & $0.52$ & $0.54$ & $0.50$ & $0.53$\\
\gmwc[128] & $2.08$ & $2.03$ & $2.05$ & $2.04$\\
\gmwc[192] & $2.05$ & $2.04$ & $2.04$ & $2.05$\\
\gmwc[256] & $2.06$ & $2.03$ & $2.05$ & $2.05$\\
\hline
MCG, $m=2^{127}-1$ & $2.19$ & \multicolumn{1}{c}{---} & $2.28$ & \multicolumn{1}{c}{---}\\
\xoroshiropp[128]~\cite{BlVSLPNG} & $0.66$ & \multicolumn{1}{c}{---} & $0.69$ & \multicolumn{1}{c}{---}\\
\xoshiropp[256]~\cite{BlVSLPNG} & $0.58$ & \multicolumn{1}{c}{---} & $0.58$ & \multicolumn{1}{c}{---}\\
RomuTrio~\cite{OvRFNPRNGPHQ} & $0.58$ & \multicolumn{1}{c}{---} & $0.59$ & \multicolumn{1}{c}{---}\\
SFC64~\cite{DotPR} & $0.49$ & \multicolumn{1}{c}{---} & $0.66$ & \multicolumn{1}{c}{---}\\
AES-CTR~\cite{SMDSPRN} & $0.68$ & \multicolumn{1}{c}{---} & $0.58$ & \multicolumn{1}{c}{---}\\
\end{tabular}
\end{table}

As with all benchmarks, the results depend on the hardware, on the compiler,
and on the way the outputs are used, but they can be explained by the
dependency chains of the code in Figure~\ref{fig:code}. Since the frequency was
fixed at $4.5$\,GHz, a nanosecond in Table~\ref{tab:speed} corresponds to $4.5$
clock cycles; on a modern superscalar processor, the latency of a full
$64\times64\to128$-bit multiplication is three or four cycles, and that of a
$128$-bit addition one or two cycles.

In \mwc[128], the multiplication by $a$ and the addition of the carry lie on the
loop-carried dependency chain, and indeed \mwc[128] needs about $4.2$ cycles per
output. For $k\geq2$, instead, $x_n$ depends on $x_{n-k}$, so the
multiplications of $k$ consecutive steps are independent and can be executed in
parallel, and the only dependency from one step to the next is the propagation
of the carry. Thus, in spite of their larger state, \mwc[192] and \mwc[256] are
almost twice as fast as \mwc[128]: they need about $2.3$ cycles per output,
close to the two cycles needed to propagate the carry (an addition followed by an
addition with carry), and they are the fastest generators in
Table~\ref{tab:speed}, except for SFC64 compiled by GCC, which, however, has a
larger state ($256$ bits) than \mwc[192], guarantees only a period of at least
$2^{64}$, and has no jump function.

In GMWC generators, the dependency chain of the carry contains also the
multiplication by $A$ that computes $x_n$ and the multiplication of $x_n$ by
$e$, which depend on each other and cannot be overlapped by increasing $k$: all
GMWC generators need about $9.2$ cycles per output, so they are about twice as
slow as \mwc[128], and four times as slow as \mwc[192] and \mwc[256]. With
AVX-512, GCC moves the $128$-bit intermediate values of the plain generators
through vector registers, which doubles their time; hence, the option
\texttt{-mno-avx512vl}.

In the MCG with modulus $2^{127}-1$ the two multiplications can be executed
in parallel, but the additions that combine their results and the reduction
lengthen the loop-carried dependency chain. The last step of the reduction, a
subtraction of $2^{127}-1$, is necessary with probability less than $2^{-63}$:
we implement it with a branch, which is never taken in practice, so it is not
part of the chain.\footnote{This implementation is due to Alexey Voskov, who uses
it in the generator \texttt{lcg127prime} of the SmokeRand test suite
(\url{https://github.com/alvoskov/SmokeRand}).} The MCG needs about $10$ cycles per output, more
than twice as much as \mwc[128], and slightly more than GMWC generators.

Finally, X adds a single instruction that is not on the loop-carried dependency
chain, and its cost is less than $7\%$; for MWC generators, computing it on the
new state instead of the current one, which shifts the output sequence by one
step, gives similar timings.

\subsection{GPUs}
\label{sec:gpu}

On GPUs, $64$-bit multiplications are emulated using $32$-bit ones, so
generators with $b=2^{32}$ become competitive. Table~\ref{tab:speed-gpu}
reports the time needed to generate a $32$-bit output, by the whole GPU and by
one of its multiprocessors, for \mwc[64X/32],
\mwc[128X/32], and \mwc[128X] (whose $64$-bit outputs count as two $32$-bit
outputs), and for SFC32~\cite{DotPR}, \xoshiropp[128]~\cite{BlVSLPNG}, and
three generators of NVIDIA's cuRAND
library:\footnote{\url{https://docs.nvidia.com/cuda/curand/}}
Philox4x32-10~\cite{SMDSPRN}, MRG32k3a~\cite{LEcGPICMRRNG}, and XORWOW, a
xorshift generator combined with a Weyl sequence~\cite{MarXR}, which is the
default generator of cuRAND.

Each thread copies the state of its generator into registers, accumulates the
outputs of $2^{24}$ calls as in the CPU harness, and writes the result to
global memory. The loop is unrolled $16$ times: GPUs issue instructions in
order, so the instructions controlling the loop, which a CPU executes in
parallel with the dependency chain of the generator, would cost as much as a
step of the fastest generators. We launch as many threads as the
multiprocessors can run at the same time, given the registers used by each
generator, as further threads would run in a second wave on partially occupied
multiprocessors, and we report the minimum time per output over five runs for
each block size in $\{128,256,512,1024\}$. For Philox, we use the function
\texttt{curand4()}, which returns four outputs at a time.

\begin{table}
\caption{\label{tab:speed-gpu}Time in picoseconds to generate a $32$-bit
output on an NVIDIA H100 SXM GPU ($132$ multiprocessors) with frequency fixed
at $1.83$\,GHz, compiled by \texttt{nvcc}~13.2 with options \texttt{-O3
-arch=sm\_90}. The first column is the time of the whole GPU, that is, the
inverse of its throughput, and the second one the time of a single
multiprocessor, the unit of a GPU analogous to a core of a CPU (the first column
multiplied by $132$).}
\renewcommand{\arraystretch}{1.1}
\small
\begin{tabular}{lrr}
Generator & \multicolumn{1}{c}{GPU} & \multicolumn{1}{c}{Multiprocessor}\\
\hline
\mwc[64X/32] & $0.23$ & $30$\\
\mwc[128X/32] & $0.24$ & $31$\\
\mwc[128X] & $0.38$ & $50$\\
\hline
SFC32~\cite{DotPR} & $0.30$ & $39$\\
\xoshiropp[128]~\cite{BlVSLPNG} & $0.44$ & $58$\\
XORWOW~\cite{MarXR} & $0.36$ & $47$\\
Philox4x32-10~\cite{SMDSPRN} & $0.81$ & $106$\\
MRG32k3a~\cite{LEcGPICMRRNG} & $1.77$ & $234$\\
\end{tabular}
\end{table}

The step of \mwc[64X/32] (Figure~\ref{fig:code32}) compiles to a single multiply--add instruction, which
computes $a x_{n-1}+c_{n-1}$ with a $64$-bit result, and a xor that computes X,
and \mwc[64X/32] is the fastest generator in Table~\ref{tab:speed-gpu}: more
than $50\%$ faster than XORWOW, which uses only additions, shifts, and xors, and
three and a half times as fast as Philox. MRG32k3a is the slowest: cuRAND
computes it using the same multiply--add instruction, but it needs several of
them per output and converts the result through a double. Unlike on CPUs, a
larger lag does not help, as the GPU hides the latency of a thread by running
other threads, and \mwc[128X/32] is slightly slower than \mwc[64X/32].
\mwc[128X], whose step needs a $64\times64\to128$-bit multiplication, is slower
per bit than \mwc[64X/32]. XORWOW, besides being slower, fails PractRand
after just $8$\,MB of output.

\subsection{Recommendations}
\label{sec:recommendations}

In all cases, we suggest using the scrambler X: it costs less than $7\%$ of the
time per output, and in our tests it removes from all bits every lattice defect
we could detect (Section~\ref{sec:x}). For $b=2^{32}$, X is
necessary: without it, \mwc[64/32], \mwc[128/32], and \gmwc[64/32] fail a
three-dimensional birthday-spacings test in less than a minute
(Section~\ref{sec:detect}).

Our suggested generator is \mwc[192X]: \mwc[192] and \mwc[256] are the fastest
generators in Table~\ref{tab:speed} with a known period and jump functions, and
\mwc[192] has the smaller state, with a period of about $2^{191}$. Its lacunary triples $(x_n,x_{n+1},x_{n+3})$ lie in planes at
distance about $2^{-32}$ (Corollary~\ref{cor:lacunary}), a defect that our
birthday-spacings test detects using about $4$\,TB of output
(Section~\ref{sec:detect}), and that X removes (Section~\ref{sec:x}); in our
opinion, however, at this state size the defect is irrelevant in practice, even without X.

GMWC generators have better lattices: \gmwc[128] shows no structural limitation in
dimensions up to eight, and a test on its worst projections needs at least $60$
times more points than a test on the successive triples of \mwc[128]
(Section~\ref{sec:detect}). However, in our tests X removes the lattice
defects of MWC generators at the cost of a single instruction, whereas GMWC
generators are about four times as slow as \mwc[192] (Section~\ref{sec:speed}):
there is thus no reason to use GMWC generators.

For $32$-bit outputs, we suggest \mwc[128X/32], which, like \mwc[192X], performs
independent multiplications, and has a period of about $2^{127}$. When registers
are scarce, as on GPUs, and a relatively short period is not an issue, we
suggest \mwc[64X/32], which is Thomas's MWC64X with our multiplier (Thomas's
multiplier is the choice of $M_8$, see Section~\ref{sec:comparison}), and is
the fastest generator in Section~\ref{sec:gpu}.

\section{Conclusions}

The form of the modulus of MWC generators reduces the figure of merit in
dimension $k+2$ to the shortest vector of a two-dimensional lattice, and bounds
it explicitly. The bound is due to dual vectors with small support, so for every
lag the vectors $(x_n,x_{n+1},x_{n+k+1})/b$ lie in planes at distance about
$b^{-1/2}$. GMWC generators with lag one avoid this limitation, at least
empirically, and those with larger lags push it to a distance of about
$b^{-2/3}$, at the price of two additional dependent multiplications. The
scrambler X removes, in our tests, the defect and every other form of lattice
structure we could detect, at the cost of a single instruction. Generators with
several nonzero coefficients can avoid the structural
limitation~\cite{CoLDPMCRNG}, but at the cost of additional multiplications, to
improve a structure that X scrambles anyway.

Since the
degenerate behavior of the figures of merit makes classical criteria unsuitable,
we select multipliers by a weighted geometric mean, which is essentially forced
by requiring that the ranking depend neither on the normalization of the figures
of merit nor on the other candidates; as every criterion based on figures of
merit, however, it says nothing about the projections it does not consider. Our
suggested generator is \mwc[192X], or \mwc[128X/32] when $32$-bit outputs are
needed (\mwc[64X/32] when registers are scarce, as on GPUs).

Several problems remain open. A theoretical analysis of X would make it possible
to compare it with GMWC generators on the same ground. It would be interesting
to show that the families of dual vectors of Proposition~\ref{prop:general} are
the only ones yielding nontrivial bounds, which would also explain the apparent
absence of a structural limitation on $f_3$ for GMWC generators with lag one; to
quantify the amount of output that a birthday-spacings test needs to detect a
given distance between hyperplanes, extending the rules of thumb of L'Ecuyer and
Simard~\cite{LESPBSTCFRNG}; and, for generators used without a scrambler, to
extend our analysis to MWC generators with several nonzero
coefficients~\cite{MarYAR,CoLDPMCRNG}.

\begin{acks}
The author would like to thank Mark Goresky for a long and very interesting
exchange on generalized multiply-with-carry generators, which included the
confirmation of the error in the formula for the carry and of the erroneous
parameter sets discussed in Sections~\ref{sec:mwc} and~\ref{sec:erratum}, and
a discussion on seeding; Pierre L'Ecuyer for useful pointers to the literature on
carry generators; David Blackman for several discussions on MWC generators; and
Raimo Niskanen for the collaboration on the Erlang/OTP \texttt{mwc59} generator.
\end{acks}

\bibliography{biblio,mwc-extra}

\end{document}